\RequirePackage{fix-cm}
\documentclass[smallextended]{svjour3}       
\smartqed  
\usepackage{graphicx}
\usepackage{balance}  
\usepackage{lipsum}
\usepackage{tabularx}
\usepackage{graphicx}
\usepackage{epstopdf}
\usepackage{psfrag}
\usepackage{color}
\usepackage{filecontents}
\usepackage{cite}
\usepackage{mathtools}

\usepackage{amssymb,amsmath,amsfonts}

\usepackage{float}
\usepackage{algorithm}
\usepackage[noend]{algpseudocode}

\usepackage{tikz}
\usepackage{ifthen}
\usepackage{subfigure}
\usepackage{overpic}
\usepackage{pict2e}

\usepackage{caption}
\usepackage{subfigure,float}
\usepackage{footnote}
\usetikzlibrary{matrix,positioning,decorations.pathreplacing}

\DeclareMathOperator{\Mcol}{\ell}
\DeclareMathOperator{\Mrow}{\ell}
\DeclareMathOperator{\Mnull}{0}

\usepackage{changepage}

\newcommand{\norm}[1]{\left\lVert#1\right\rVert}
\newcommand{\algo}{{\sc LRC-Katz}}
\newcommand{\algoLink}{{\sc Sparse-Katz}}

\newcommand{\boldG}{\mathbf{G}}
\newcommand{\Katzmatrix}{\mathbf{K}}
\newcommand{\Identitymatrix}{\mathbf{I}}
\newcommand{\Katzvector}{\mathbf{k}_q}              
\newcommand{\BigM}{\mathbf{M}}                      
\newcommand{\BigS}{\mathbf{S}}                      
\newcommand{\network}{\cal{G}}
\newcommand{\nodeset}{\cal{V}}
\newcommand{\edgeset}{\cal{E}}
\newcommand{\netW}{\mathbf{W}}  
\newcommand{\PositiveSet}{{\cal{P}}}

\newcommand{\K}{{k}}

\newcommand{\residualArnoldi}{{\mathbf{l}}}

\newcommand{\BigMeasy}{{\textbf{$\BigM_{11}$}}}
\newcommand{\BigMright}{{\textbf{$\BigM_{12}$}}}
\newcommand{\BigMleft}{{\textbf{$\BigM_{12}^T$}}}
\newcommand{\BigMhard}{{\textbf{$\BigM_{22}$}}}

\newcommand{\SchurComplement}{{\textbf{${\mathbf{S}} = \BigMhard - \BigMleft {\mathbf{M}}_{11}^{-1} \BigMright$}}}

\begin{document}

\title{Fast Computation of Katz Index for Efficient Processing of Link Prediction Queries
}

\titlerunning{Fast Katz for Efficient Link Prediction Queries}        

\author{Mustafa Co\c{s}kun \and
        Abdelkader Baggag  \and
        Mehmet Koyut\"urk
}


\institute{Mustafa Co\c{s}kun \at
              Department of Computer Engineering\\
              Abdullah G\"ul University \\
              Kayseri,Turkey\\
              Tel.: +905-05-0082739\\
              \email{mustafa.coskun@agu.edu.tr}           
           \and
           Abdelkader Baggag \at
              Qatar Computing Research Institute\\
              Hamad Bin Khalifa University \\
              Doha, Qatar\\
              Tel.: +971-4454-7250\\
              \email{abaggag@hbku.edu.qa}
          \and
         Mehmet Koyut\"urk \at
         Department of Electrical Engineering and Computer Science\\
         Case Western Reserve University\\
         Cleveland, USA\\
         Tel.: +121-63-682963\\
         \email{mehmet.koyuturk@case.edu}
}

\date{Received: date / Accepted: date}

\maketitle
\begin{abstract}
Network proximity computations are among the most common operations in various data mining applications, including link prediction and collaborative filtering. 
A common measure of network proximity is Katz index, which has been
shown to be among the best-performing path-based link prediction algorithms. 
With the emergence of very large network databases, such proximity computations become an important part of query processing in these databases.
Consequently, significant effort has been devoted to developing algorithms for efficient computation of Katz index between a given pair of nodes or between a query node and every other node in the network. 
%
Here, we present \algo, an algorithm based on indexing and low rank correction to accelerate Katz index based network proximity queries. 
Using a variety of very large real-world networks, we show that \algo\
outperforms the fastest existing method, Conjugate Gradient, for a wide range of parameter values. 
We also show that, this acceleration in the computation of Katz index can be used to drastically improve the efficiency of processing link prediction queries in very large networks. 
Motivated by this observation, we propose a new link prediction algorithm that exploits modularity of networks that are encountered in practical applications. 
Our experimental results on the link prediction problem shows that our modularity based algorithm  significantly outperforms state-of-the-art link prediction Katz method.
\end{abstract}
\section{Introduction}
\label{sec:introduction}
Proximity measure computation for nodes in networks is a well-adapted operation in many data analytic
applications. In link prediction or recommender systems, network proximity measures the node
similarity in social networks ~\cite{LinkLiben,RecommedationSysSarkar}. In information retrieval, anomalous link are ranked based on nodes proximity to other
nodes~\cite{ALD}. In unsupervised learning, the network proximity quantify a cluster quality ~\cite{clusters}.

The general setting for network proximity queries is as follows: Given a query node and some network distance measurement, we aim at computing a score for all other nodes in the network, based on their
proximity with respect to the query node. For instance, these
measures of network distance include shortest path, which computes the minimum number of edges between two nodes, and Katz-based proximity, which can be defined as a measure in between any pair of nodes in a network that capture their relationship due to the various paths connecting these two pair of nodes. In many applications, Katz-based proximity are preferable to shortest path distance, since it captures the
global structure of a network and multi-paths relationships of the
nodes~\cite{Katz,FastKatz}. Katz-based proximity measures
have been used in a number of applications, including link prediction ~\cite{LinkLiben},
clustering~\cite{clusters}, and ranking \cite{ALD}.

Motivated by social network data mining problems for instance
link prediction and collaborative filtering, significant efforts have been devoted to reducing computational costs associated Katz-based proximity. These efforts typically speed-up computations
by taking one of the following approaches: (i) exploiting numerical properties of iterative methods, along with structural characteristics of the underlying networks 
to speed up query processing; (ii) avoiding iterative computations during query processing by inverting the underlying system of equations using Cholesky factorization and storing
the resulting factorization as an index. For instance, in the context of top-$k$ Katz-based proximity queries, the state-of-the-art
methods~\cite{FastKatz } use breadth-first
ordering of the nodes in the network to bound element-wise increments of proximity scores by exploiting a relationship between
the Lanczos process and a quadrature rule in the iterative computation \cite{FastKatz}. This bounding process eliminates nodes whose proximity values cannot exceed those of the nodes that are already among the top-$\K$ most proximate to the query node by only entering some part of the underlying network. Likewise, using Cholesky factorization of the underlying linear system, top-$\K$ proximity computations can be computed efficiently~\cite{SaadBook}. These approaches have been demonstrated
to yield significant improvement in computation time, however, their application to larger networks is limited. Specifically, for very large networks, iterative methods~\cite{FastKatz, SaadBook}
require a large number of iterations to converge. More concretely, Fast-Katz \cite{FastKatz} offers a tight 
convergence upper bound ,however, Conjugate Gradient (CG) performs better than Fast-Katz \cite{FastKatz} for Katz-based proximity. On the other hand, direct inversion
techniques ~\cite{SaadBook} are not scalable to large matrices,
since the inverse of a sparse matrix is usually dense.

In this paper, we propose a hybrid approach that partitions the network into disjoint sub-networks and inverts the small matrices corresponding to these subnetworks. The denser matrix,
composed of the nodes connecting these subnetworks is handled through a low rank corrected iterative CG method. By inverting small block diagonal matrices, the hybrid procedure
overcomes the memory constraint of direct inversion techniques. By performing an iterative low rank corrected
procedure on a smaller matrix, it overcomes
the computational cost considerations of classical CG method. We first use graph partitioning ~\cite{karypis1998parallel} to partition the matrix into a sparse
block-diagonal matrix and a dense, but much smaller matrix. And, we index the inverse of the sparse block diagonal matrix through a computationally efficient
procedure. We then propose an efficient low rank correction technique along with CG, to speed up the iterative process for the dense matrix.
At query time of Katz-based proximity, we first solve the dense part of the system using the low rank corrected iterative solver,
and use the stored index (of the block diagonal matrix) to solve the rest of the system.

We describe these processes in detail, and show that the resulting hybrid approach converges much faster than classical CG. Our hybrid approach significantly accelerates the computation of
Katz-based proximity queries for very large graphs.
We provide a detailed theoretical justifications for our results and experimentally show the superior performance of our method on a number of real-world networks.
Specifically, we show that our method yields over at least 3-fold improvement in the runtime of online query processing
over the best state-of-the-art method, CG across all our experiments. 

Moreover, we also show that, low rank corrected Katz-based proximity can be used to drastically improve the efficiency of processing link prediction queries in very large networks. 
Motivated by skewed distribution of Katz scores \cite{FastKatz}, we propose a new link prediction algorithm, called \algoLink\, that exploits modularity of networks that are encountered in practical applications. 
Our experimental results on the link prediction problem shows that our modularity based algorithm  significantly outperforms state-of-the-art link prediction Katz method.

In summary, the two main contributions of the proposed framework are the following:
\begin{itemize}
\item We introduce a hybrid approach to indexing-based acceleration of Katz-based network proximity queries, in which the network is divided into two components, where the larger and sparser part of the resulting system is solved by indexing the inverse of the corresponding matrix, and the smaller and denser part of the system is solved during query processing using proposed low rank corrected CG method. 
\item	We introduce another link prediction algorithm that renders Katz-based link prediction approach more effective by exploiting the skewed distribution of Katz scores.
\end{itemize}

Taken together, these two contributions bring the field closer to real-time processing of proximity queries as well as link prediction task on very large networks.

The rest of the paper is organized as follows: in the next section, we provide a
review of the literature on efficient processing of Katz-based network proximity and link prediction.
In Section~\ref{sec:methods}, we define Katz-based proximity and link prediction problem, and describe our method, along with its theoretical justifications. In Section~\ref{sec:results},
we provide detailed experimental assessments of our method on very large networks for both Katz based proximity and its link prediction task.
We conclude our discussion and summarize avenues for future research in Section~\ref{sec:conclusion}.

\section{Related Work}
\label{sec:related}
Node proximity queries have been soaring significant research attention in recent years in the
context of searching, ranking, clustering and analyzing network structured object similarity~\cite{coskun18}. In particular, Katz-based node proximity queries in
large graph have been well studied ~\cite{FastKatz}.
One of the commonly used approaches to computing Katz-base based proximity is the power method through the Neumann series expansion of the underlying linear systems of equations ~\cite{SaadBook}.

An alternate approach to power iterations is the use of offline computation, which directly inverts the underlying linear system of equations, typically using Cholesky factorization or
eigen-decomposition~\cite{Acar, Tracable, Wang}. These methods
tend compute network proximity rapidly, using a single matrix vector multiplication,
however, they involve in some expensive preprocessing, and their memory requirements constrain
their use to smaller networks.In contrast, we here propose to compute an approximate solution to linear system of equations
by approximating the inverse of Schur Matrix with low rank correction approach, and use
this low rank approximation as a preconditioner in CG to compute the exact solution at query time. 
In doing so, we significantly improve on both memory and storage requirements for indexing as well as runtime of proximity computation.

There have also been extensive efforts aimed at scaling top-$\K$ proximity queries to large sparse
networks for Katz-based proximity. These methods
utilize the topology of the network to perform a local search around the query node, based on the premise that nodes with high random-walk based
proximity to the query node are also close to the query node in terms of the number of hops by exploiting a relationship between
the Lanczos process and a quadrature rule in the iterative computation ~\cite{FastKatz}.
However, these local search based methods for Katz proximity computation are not as efficient as CG method.

In this paper, we focus on exact computation of Katz proximity in very large networks using a novel, hybrid approach accelerated with low rank correction. Our method is fundamentally different from existing approaches in that it simultaneously targets scalability and efficiency. Our method can be used for efficiently computing
Katz-based proximity of all nodes in a undirected network. Moreover, we also develop a new link prediction algorithm that takes advantages of sparse distribution Katz scores and drastically improves the effectiveness of Katz-based link prediction approach.

\section{Methods}
\label{sec:methods}



In this section, we first define Katz-index and formulate node proximity queries based on Katz index.
We then describe our approach to indexing, which is based on graph-partitioning indexing, i.e., to partition the resulting linear system and to index the sparser part of the system. 
Subsequently, we discuss how the iterative computation can be accelerated using low rank correction to refine the solution, and solve the remaining part of the linear system. 
Finally, we discuss how these two approaches can be used in combination, to efficiently process Katz-based network proximity queries. 
The workflow of the proposed framework is shown in Figure~\ref{fig:mainFigure}.

\subsection{Katz Index}
Let ${\network} ={({\nodeset},{\edgeset})}$ be an undirected and connected network, where 
$\nodeset$ denotes the set of $ | \nodeset | $ nodes and $\edgeset$ denotes the set of edges, with sizes indicated as $ | \nodeset | $ and $ | \edgeset | $, respectively.
Katz index quantifies the proximity between a pair of nodes in this network as the weighted sum of 
all paths connecting the two nodes, where the weights of the paths decay exponentially with path length~\cite{Katz,FastKatz}.  
Namely, for a pair of nodes $i$ and $j \in \nodeset$, the Katz index is defined as:
\begin{equation}
K_{i,j} = \sum_{l=1}^{\infty} \alpha^{l} {paths}_{l}(i,j) ,
\label{eq:KatzPath}
\end{equation}
where ${paths}_{l}(i,j)$ denotes the number of paths of length $l$ connecting  $i$ and $j$ in $\network$.
The parameter $\alpha$ is a damping factor that is used to tune the relative importance of longer paths, where $0 < \alpha < 1$, thus smaller $\alpha$ assigns more importance to shorter paths~\cite{FastKatz}. 

 By observing the relationship between the number of paths in $\network$ and the powers of the adjacency matrix $\boldG$ of $\network$, of size $ | \nodeset | \times | \nodeset | $, the computation of Katz index can be formulated as an algebraic problem. 
 To see this, let $\Katzmatrix$ denote the \textit{Katz matrix}, i.e., the matrix of Katz indices between all pairs of nodes in $\network$. 
 
 Since $({\boldG}^{l})_{i,j}$ is equal to the number of paths of length $l$ between $i$ and $j$,  $\Katzmatrix$ can be written as:
\begin{equation}
{\Katzmatrix} = \sum_{l=1}^{\infty}{\alpha^{l} {\boldG}^{l}} = {(\Identitymatrix - \alpha {\boldG})}^{-1} - {\Identitymatrix} ,
\label{eq:KatzMatrix}
\end{equation}
where $\Identitymatrix$ is the identity matrix. 
In the rest of our discussion, we assume that $\alpha < 1/ {\norm{\boldG}}_2$   
to ensure that ${(\Identitymatrix - \alpha {\boldG})}$ is positive definite, and to guarantee the convergence of the Neumann series to the inverse of ${(\Identitymatrix - \alpha {\boldG})}$. 

For a given query node $q$, the Katz index of every other node in $\network$ with
respect to $q$ is given by the $q$th column of $\Katzmatrix$, which we denote $\Katzvector$. 
Observe that the computation of $\Katzvector$ corresponds to solving the following  linear system:
\begin{equation}
\BigM  \Katzvector = \alpha {\mathbf{g}}_q. 
\label{eq:LinearSystem}
\end{equation}
Here, $\BigM = (\Identitymatrix - \alpha \boldG)$ and 
${\mathbf{g}}_q$ is the $q$th column of the adjacency matrix $\boldG$.
Since $\BigM$ does not depend on the query node $q$, it can be inverted 
offline, and the inverse can be used as an index to compute ${\Katzvector} = \alpha \BigM^{-1} {\mathbf{g}}_q$ by performing a single matrix vector multiplication
during query processing.
However, inverting $\BigM$ and storing ${\BigM}^{-1}$ is not feasible for very large graphs, since the inverse of a general sparse matrix is dense. 
Therefore, existing algorithms solve this linear system of equations using an iterative solver such as the (preconditioned) Conjugate Gradient method~\cite[p. 196]{SaadBook}, which is applicable in this case since $\BigM$ is symmetric ~\cite{FastKatz}. 
While these iterative methods greatly accelerate the computation of Katz index, they are not fast enough to enable real-time query processing in very large graphs.

Recently, to enable efficient processing of random-walk based queries on billion-scale networks, we have developed {\sc I-Chopper}, a hybrid method that uses a combination of 
indexing and accelerated iterative solvers~\cite{coskun18}.
In the following subsection, we show how a similar idea can be applied to the processing of Katz index based queries by indexing the inverse of the matrix that corresponds to sparser parts of the network and performing low-rank correction at query time to obtain an exact solution for the rest of the network. 
We then describe how the resulting algorithm, \algo, can be used to efficiently perform
Katz index based link prediction.

\subsection{Graph Partitioning Based Indexing}
As in {\sc I-Chopper}, \algo\ exploits the sparsity of real-world networks to efficiently compute and index the inverse of a large part of $\BigM$. 
The key insight behind this approach is that, although the inverse of a general sparse matrix is dense, the inverse of a block-diagonal matrix with a low bandwidth is sparse~[]p.87]\cite{SaadBook}.

Since most real-world networks are scale-free, most of the nodes in the network have low degree.
Consequently, observing that the non-zero structure of the matrix $\BigM$ is identical to that of the adjacency matrix of the network, we can reorder the rows of $\BigM$ such that $\BigM$ can be partitioned into a very large block-diagonal matrix with a low bandwidth (corresponding to small connected subgraphs consisting of low-degree nodes) and relatively dense but much smaller matrices (corresponding to hubs and their connections to other nodes).
%
The following lemma shows how such partitioning of $\BigM$ can be used to separate the solution of the linear system  of Equation~\ref{eq:LinearSystem} into two parts:

\begin{lemma}
Suppose a linear system $\BigM {\mathbf{k}}_q = {\tilde{\mathbf{g}}}_q$, can be partitioned as 
\begin{equation}
\begin{bmatrix}
 \BigMeasy & \BigMright \\ 
 \BigMleft & \BigMhard
\end{bmatrix}
\begin{bmatrix}
 \mathbf{k}_{q1} \\ \mathbf{k}_{q2}
\end{bmatrix}
= 
\begin{bmatrix}
{\tilde{\mathbf{g}}}_{q1} \\ {\tilde{\mathbf{g}}}_{q2}
\end{bmatrix}    , 
\end{equation}
such that $\BigMeasy$ is invertible. Letting $\SchurComplement$
denote the Schur complement, 
the linear system can be solved as:
\begin{align}
    \mathbf{k}_{q2} &= 
    {\mathbf{S}}^{-1} \left( {\tilde{\mathbf{g}}}_{q2} -  \BigMleft {\mathbf{M}}_{11}^{-1} {\tilde{\mathbf{g}}}_{q1} \right) , \\
    \mathbf{k}_{q1} &= 
    {\mathbf{M}}_{11}^{-1} \left( {\tilde{\mathbf{g}}}_{q1} - \BigMright {\mathbf{k}_{q2}} \right) .
\end{align}
%
%
\label{ref:blockLemma}
\end{lemma}
\begin{proof}
The proof of this lemma is straightforward, and hence it is not included.
\end{proof}

In our application, ${\tilde{\mathbf{g}}}_q=\alpha {\mathbf{g}}_q$.
This lemma applies to the computation of Katz indices as long as $\alpha< 1/ {\norm{\boldG}}_2$,
since $\BigM$ is diagonally dominant and invertible in that case. 
In the light of this lemma, the computation of Katz indices can be performed as follows:
\begin{enumerate}
    \item {\bf Indexing}: Construct $\BigM$. 
    \item {\bf Indexing}: Partition $\network$ using multi-way minimum-vertex-separator partitioning.
    \item {\bf Indexing}: Reorder $\BigM$ so that $\BigM_{11}$ contains the internal edges of all parts resulting from the partitioning with nodes within each partition corresponding to successive rows (hence columns), $\BigMright=\BigMleft$ contains the edges between nodes in partitions and nodes in the separator, and $\BigM_{22}$ contains the edges between nodes in the separator. 
    \item {\bf Indexing}: Compute and store $\mathbf{M}_{11}^{-1}$, $\BigMright$, and $\BigS$.
    \item {\bf Query Processing}: For a query $q$, compute $\mathbf{k}_{q2}$ as given in Equation (5), but without inverting $\BigS$, as described below.
    \item {\bf Query Processing}: Compute $\mathbf{k}_{q1}$ by performing two matrix-vector multiplications as given in Equation (6).
\end{enumerate}

This procedure is identical to the procedure implemented in {\sc I-Chopper}, with one important 
difference in the computation of $\mathbf{k}_{q2}$ during query processing (Step 5).
In {\sc I-Chopper}, this computation is accelerated using Chebyshev polynomials over the elliptic plane~\cite{coskun18}.
In the computation of Katz-index, $\BigS$ is symmetric~[p.271]\cite{SaadBook}, thus the elliptic plane degrades to the real axis and the solution can be found by using Chebyshev polynomials on the real axes~\cite{coskun2016efficient}.
However, this requires knowledge of the largest and smallest eigenvalues of $\BigS$, and it is
costly to compute these eigenvalues.
Since the matrix in question is not symmetric in queries involving random walks, {\sc I-Chopper}
addresses this problem by pre-computing these eigenvalues using Arnoldi's method ~[p. 160]\cite{SaadBook}.
In the computation of Katz-index, however, $\BigS$ is symmetric, and therefore this computation can be avoided by utilizing methods that do not require an eigen-bound.  
Specifically, we use the Conjugate Gradient (CG) method to solve the linear system involving $\BigS$, since CG does not require the knowledge of the largest and smallest eigenvalues of $\BigS$.
Still, this computation can be accelerated using eigenvectors of a low-rank approximation of 
$\BigS^{-1}$. 

Since the details of all other steps are described in \cite{coskun18}, here we briefly describe the key idea in each step. 
We then focus on the description of low-rank approximation and describe this approach in detail.

{\bf Steps 2 and 3 -- Partitioning of $\network$ and Reordering of $\BigM$:} The idea behind the partitioning of the network is to 
reorder the rows and columns of $\BigM$ in such a way that we can obtain a block diagonal 
$\BigMeasy$ with a small bandwidth, i.e., the non-zero entries in matrix $\BigMeasy$ are condensed
around its diagonal, ensuring that ${\mathbf{M}}_{11}^{-1}$ is sparse. 
To accomplish this, we use multi-way minimum-vertex-cover partitioning to partition the nodes of $\network$ into $p$ partitions such that each node in partition $\Pi_i$ are connected only to nodes in $\Pi_i$ or to a set $\Pi_s$ of nodes that are classified as the ``vertex-separator''~\cite{karypis1998fast,karypis1998parallel}.
Given such a partitioning, we reorder the matrix $\BigM$ such that the rows/columns that correspond to nodes in the partitions $\Pi_1,  \Pi_2, ..., \Pi_p$ are ordered next to each other, and
rows/columns that correspond to the nodes in $\Pi_s$ are at the bottom/right of the matrix. 
As a result, the reordered matrix $\BigM$ can be divided into the following sub-matrices: (i) $\BigMeasy$
contains the non-zeros that correspond to the edges within the partitions, 
(ii) $\BigM_{12}$ and ${\mathbf{M}}_{12}^T$ contain the non-zeros that correspond to the edges between nodes in
a partition and nodes in $\Pi_s$, (iii) $\BigM_{22}$ contains the non-zeros that correspond to the edges between nodes $\Pi_s$.  
Since minimum-vertex-separator graph partioning is a NP-hard problem, we use a heuristic that is well-suited to our application.
Namely, the {\sc Part-GraphRecursive} package implemented in the MeTiS graph
partitioning tool [21] allows the user to put a threshold on the size of the vertex separator, 
as opposed to minimizing it,
and recursively bipartitions the network until this threshold is reached.
Therefore, we can directly control the size of $\BigS$ (number of rows/columns of $\BigS$ is equal
to the number of nodes in the vertex separator) and the recursive partitioning generates many small 
partitions with roughly equal sizes, thereby keeping the bandwidth of $\BigMeasy$ small. 



{\bf Step 4 -- Computation of ${\mathbf{M}}_{11}^{-1}$ and $\BigS$.}
Once $\BigMeasy$ is constructed, we invert ${\mathbf{M}}_{11}^{-1}$, which is also relatively sparse
and can be stored as an index. Here, we remark that the inversion of $\BigMeasy$ is feasible even for {\emph{graphs with hundred millions of nodes}} since it is block diagonal with a small bandwidth and there exists many efficient algorithms for inverting banded matrices \cite{coskun18}. 
In our implementation, we use the Incomplete Cholesky factorization along with approximate minimum algorithms~\cite{amd1,amd2} before we invert the sparse block diagonal matrix $\BigMeasy$. 
Once ${\mathbf{M}}_{11}^{-1}$ is available, we compute $\BigS$ as defined in Lemma 1,
and store ${\mathbf{M}}_{11}^{-1}$, $\BigS$, and $\BigMright$.

\begin{figure*}[t]
	\centering
	\scalebox{0.42}{\input{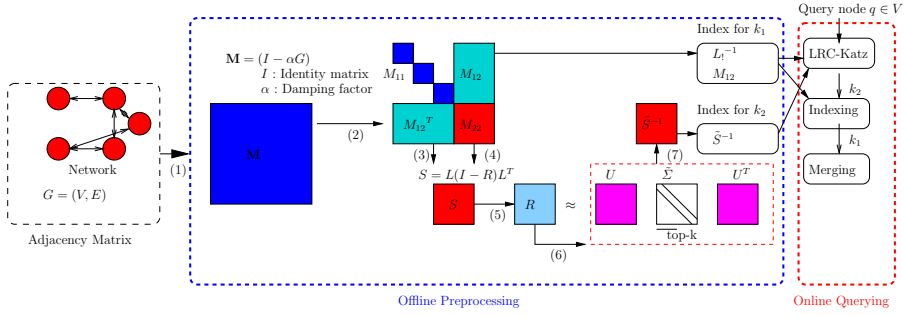}}
	\caption{{\bf Flowchart illustrating the proposed framework for indexed processing
of Katz-based proximity queries on large networks.} 
} 
\label{fig:mainFigure}
\end{figure*}

{\bf Step 5 -- Computation of ${\mathbf{k}}_{q2}$ During Query Processing.}

Recall that processing of a Katz index query involves the 
computation of ${\mathbf{k}}_{q}$ for a given query node $q$.
As described in Lemma 1, we divide the computation of ${\mathbf{k}}_{q}$ into
the computation of ${\mathbf{k}}_{q1}$ and the computation of ${\mathbf{k}}_{q2}$.
Since the computation of ${\mathbf{k}}_{q1}$ required knowledge of ${\mathbf{k}}_{q2}$,
we first compute ${\mathbf{k}}_{q2}$ during query processing. 
This computation requires solution of the system
\begin{equation}
\BigS {\mathbf{k}}_{q2} = ( {\mathbf{g}}_2 - {\mathbf{M}}_{12}^T {\mathbf{M}}_{11}^{-1} {\mathbf{g}}_1) = {\mathbf{f}} ,
\label{eq:SKatz}
\end{equation}
where  ${\mathbf{f}}$ can be computed efficiently (by performing a single matrix-vector multiplication) during query processing, since we form and index ${\mathbf{M}}_{12}^T$ and
${\mathbf{M}}_{11}^{-1}$ in Steps 3 and 4. 
However, solving the linear system $\BigS {\mathbf{k}}_{q2} = {\mathbf{f}}$,
during query processing or pre-computing and storing the inverse of $\BigS$ is not 
feasible since $\BigS$ is a relatively dense matrix.
For this reason, we compute a low-rank approximation for $\BigS$ offline and 
store this approximation as an index that can be used to efficiently compute ${\mathbf{k}}_{q2}$ 
during query processing. 
We now explain this process. 
To avoid cluttered notation, we drop the subscripts $(q)$ in the following sections.

\subsubsection{Low Rank Correction}
The idea behind Low Rank Correct Katz Algorithm (\algo) is as follows:  
To solve $\BigS {\mathbf{k}}_2 = {\mathbf{f}}$, we approximate the Schur complement
$\BigS \in \mathbb{R}^{n_2 \times n_2}$ via $\BigMhard$ plus some low rank vectors so that we use sparser matrices instead of dense matrix $\BigS$. 

Let $\BigMhard = {\mathbf{L}} {\mathbf{L}}^T$ be the Cholesky factorization of $\BigMhard$ and recall that the Schur complement matrix can be rewritten as 
\begin{align}
{\mathbf{S}} &= {\mathbf{L}} {\mathbf{L}}^T - {\mathbf{M}}_{12}^T {\mathbf{M}}^{-1}_{11} {\mathbf{M}}_{12} 
\\
& = {\mathbf{L}} ( {\mathbf{I}} - {\mathbf{L}}^{-1} {\mathbf{M}}_{12}^T {\mathbf{M}}^{-1}_{11} {\mathbf{M}}_{12} {\mathbf{L}}^{-T} ) {\mathbf{L}}^{T} 
\\
& = {\mathbf{L}} ( {\mathbf{I}} - {\mathbf{R}} ) {\mathbf{L}}^{T}
\end{align}
%
Now define the eigen-decomposition of the symmetric matrix ${\mathbf{R}}$ as follows:
\begin{equation}
{\mathbf{R}} = {\mathbf{L}}^{-1} {\mathbf{M}}_{12}^T {\mathbf{M}}^{-1}_{11} {\mathbf{M}}_{12} {\mathbf{L}}^{-T}
= {\mathbf{U}} {\boldsymbol{\Sigma}} {\mathbf{U}}^T ,
\label{eq:RMatrix} 
\end{equation}
where the diagonal entries of ${\boldsymbol{\Sigma}}$ are the eigenvalues of ${\mathbf{R}}$ and ${\mathbf{U}}$ is the column matrix that contains the corresponding eigenvectors, which are orthogonal to each other. Then $\BigS$ can be rewritten as: 
\begin{equation}
 \BigS = {\mathbf{L}} ( {\mathbf{I}} - {\mathbf{R}} ) {\mathbf{L}}^{T} 
       = {\mathbf{L}} ( {\mathbf{I}} - {\mathbf{U}} {\boldsymbol{\Sigma}} {\mathbf{U}}^T ) {\mathbf{L}}^T 
       = {\mathbf{L}} {\mathbf{U}} ( {\mathbf{I}} - {\boldsymbol{\Sigma}} ) {( {\mathbf{L}} {\mathbf{U}} )}^T .
\end{equation}
%
%

Thus,  the inverse of the Schur complement matrix $\BigS$ becomes:
\begin{equation}
\begin{aligned}
\BigS^{-1} & =
{ \left( {\mathbf{L}} {\mathbf{U}} ( {\mathbf{I}} - {\boldsymbol{\Sigma}} ) {( {\mathbf{L}} {\mathbf{U}} )}^T \right) }^{-1}
\\
& = {\mathbf{L}}^{-T} {\mathbf{U}} { ( {\mathbf{I}} - {\boldsymbol{\Sigma}} ) }^{-1} {\mathbf{U}}^T {\mathbf{L}}^{-1} 
\\
& = {\mathbf{L}}^{-T} \left[ {\mathbf{I}} + {\mathbf{U}} {( {\mathbf{I}} - {\boldsymbol{\Sigma}} )}^{-1} {\mathbf{U}}^T - {\mathbf{I}} \right] {\mathbf{L}}^{-1} \\
& = {\mathbf{M}}_{22}^{-1} + {\mathbf{L}}^{-T} {\mathbf{U}} \left[ {( {\mathbf{I}} - {\boldsymbol{\Sigma}} )}^{-1} - {\mathbf{I}} \right] {\mathbf{U}}^T {\mathbf{L}}^{-1}.
\end{aligned}
\label{eq:inverseS}
\end{equation}
%

Now consider approximating ${\mathbf{R}}$ using its most dominant $\ell$ eigenvectors.
That is, define $\Tilde{R} \approx \tilde{\mathbf{U}} \Tilde{\Sigma} {\tilde{\mathbf{U}}}^T$, 
where $\tilde{\mathbf{U}}$ and $\Tilde{\Sigma} \in \mathbb{R}^{n_2 \times n_2}$, $diag (\Tilde{\Sigma}) =  (\sigma_1, \sigma_2, .., \sigma_{\ell}, 0, 0, ..., 0)$ and $\tilde{\mathbf{U}}$ consists of first $\ell$ eigenvectors of ${\mathbf{U}}$ padded with zeros, i.e.,:

\begin{align*}
{\mathbf{R}} & \approx \tilde{\mathbf{U}} \Tilde{\Sigma} {\tilde{\mathbf{U}}}^T \\
&=
\begin{tikzpicture}[
baseline,
mymat/.style={
  matrix of math nodes,
  ampersand replacement=\&,
  left delimiter=(,
  right delimiter=),
  nodes in empty cells,
  nodes={outer sep=-\pgflinewidth,text depth=0.5ex,text height=1.2ex,text width=0.6em}
}
]
\begin{scope}[every right delimiter/.style={xshift=-3ex}]
\matrix[mymat] (matu)
{
 \& \& \& \& \& \\
\& \& \& \& \& \\
\& \& \& \& \& \\
\& \& \& \& \& \\
\& \& \& \& \& \\
\& \& \& \& \& \\
};
\node 
  at ([shift={(3pt,-7pt)}]matu-3-2.west) 
  {$\cdots$};
\node 
  at ([shift={(3pt,-7pt)}]matu-3-5.west) 
  {$\cdots$};
\foreach \Columna/\Valor in {1/1,3/l,4/{l+1},6/n_2}
{
\draw 
  (matu-1-\Columna.north west)
    rectangle
  ([xshift=4pt]matu-6-\Columna.south west);
\node[above] 
  at ([xshift=2pt]matu-1-\Columna.north west) 
  {$u_{\Valor}$};
}
\draw[decorate,decoration={brace,mirror,raise=3pt}] 
  (matu-6-1.south west) -- 
   node[below=4pt] {$\Mcol$}
  ([xshift=4pt]matu-6-3.south west);
\draw[decorate,decoration={brace,mirror,raise=3pt}] 
  (matu-6-4.south west) -- 
   node[below=4pt] {$\Mnull$}
  ([xshift=4pt]matu-6-6.south west);
\end{scope}
\matrix[mymat,right=10pt of matu] (matsigma)
{
\sigma_{1} \& \& \& \& \& \\
\& \ddots \& \& \& \& \\
\& \& \sigma_{l} \& \& \& \\
\& \& \& 0 \& \& \\
\& \& \& \& \ddots \& \\
\& \& \& \& \& 0 \\
};
\matrix[mymat,right=25pt of matsigma] (matv)
{
 \& \& \& \& \& \\
\& \& \& \& \& \\
\& \& \& \& \& \\
\& \& \& \& \& \\
\& \& \& \& \& \\
\& \& \& \& \& \\
};
\foreach \Fila/\Valor in {1/1,3/l,4/{l+1},6/n_2}
{
\draw 
  ([yshift=-6pt]matv-\Fila-1.north west)
    rectangle
  ([yshift=-10pt]matv-\Fila-6.north east);
\node[right=12pt] 
  at ([yshift=-8pt]matv-\Fila-6.north east) 
  {$u^{T}_{\Valor}$};
}
\draw[decorate,decoration={brace,raise=37pt}] 
  ([yshift=-6pt]matv-1-6.north east) -- 
   node[right=38pt] {$\Mrow$}
  ([yshift=-10pt]matv-3-6.north east);
\draw[decorate,decoration={brace,raise=37pt}] 
  ([yshift=-6pt]matv-4-6.north east) -- 
   node[right=38pt] {$\Mnull$}
  ([yshift=-10pt]matv-6-6.north east);
\end{tikzpicture}
\end{align*}
Using this approximation to $\mathbf{R}$, we define an approximation  to $\BigS^{-1}$ as follows:
\begin{equation}
\Tilde{\BigS}^{-1}= \BigMhard^{-1} + \mathbf{L}^{-T} \tilde{\mathbf{U}} [{(\mathbf{I} - \tilde{\mathbf{\Sigma}})}^{-1} -\mathbf{I}] \tilde{\mathbf{U}}^T \mathbf{L}^{-1}
\label{eq:inverseTildeS}
\end{equation}

Note that we never compute $\Tilde{\BigS}^{-1}$ in practice, we define it here solely
for theoretical justification.


The following theorem establishes the relationship between the eigenvalues
of $\BigS \Tilde{\BigS}^{-1}$ and the eigenvalues of ${\mathbf{R}}$.

\begin{theorem}
Assume that the eigenvalues of $\mathbf{R}$ are ordered as $\sigma_1 \geq \sigma_2 \geq .... \geq \sigma_{n_2}$, where $n_2$ is size of $\BigS$. For a given integer $\ell$, define $\Tilde{\BigS}^{-1}$
as in Equation~(\ref{eq:inverseTildeS}).
Then, the eigenvalues of $\BigS \Tilde{\BigS}^{-1}$ are in the form of
\[
  \lambda_i =
  \begin{cases}
                                   1 & \text{if $i\leq \ell$} .\\
                                     1 - \sigma_i & \text{otherwise} 
  \end{cases}
\]
\label{ref:MainTheorem}
\end{theorem}

\begin{proof}
From Equations~\eqref{eq:inverseS} and~\eqref{eq:inverseTildeS}, we can write $\BigS^{-1} - \Tilde{\BigS}^{-1} = \mathbf{L}^{-T} \tilde{\mathbf{U}} [{(\mathbf{I} - \mathbf{\Sigma})}^{-1} - {(\mathbf{I} - \mathbf{\Tilde{\Sigma}})}^{-1}] \tilde{\mathbf{U}}^T \mathbf{L}^{-1} $. Then, we have, $$\mathbf{L}^T \BigS^{-1} \mathbf{L} - \mathbf{L}^T \Tilde{\BigS}^{-1} \mathbf{L} =  \tilde{\mathbf{U}} [{(\mathbf{I} - \mathbf{\Sigma})}^{-1} - {(\mathbf{I} - \tilde{\mathbf{\Sigma}})}^{-1}] \mathbf{U}^T.$$ 
Multiplying both sides of the above equality with ${(\mathbf{L}^T \BigS^{-1} \mathbf{L})\tilde{\mathbf{U}}}^{-1} = \mathbf{L}^{-1} \BigS \mathbf{L}^{-T}$, we have 
$$ \mathbf{I} - \mathbf{L}^{-1} \BigS \Tilde{\BigS}^{-1} \mathbf{L} = \mathbf{L}^{-1} \BigS \mathbf{L}^{-T} \tilde{\mathbf{U}} [{(\mathbf{I} - \mathbf{\Sigma})}^{-1} - {(\mathbf{I} - \tilde{\mathbf{\Sigma}})}^{-1}] \tilde{\mathbf{U}}^T.$$
From the definition, we know that $\mathbf{L}^{-1} \BigS \mathbf{L}^{-T} = (\mathbf{I} - \mathbf{R}) = \tilde{\mathbf{U}} (\mathbf{I} - \mathbf{\Sigma}) \tilde{\mathbf{U}}^T$. Then, by using orthogonality of $\mathbf{U}$, we have

$\begin{aligned}
 &\mathbf{I} - \mathbf{L}^{-1} \BigS \Tilde{\BigS}^{-1} \mathbf{L} = \\
 &\tilde{\mathbf{U}} (\mathbf{I} -\mathbf{\Sigma })\tilde{\mathbf{U}}^T \tilde{\mathbf{U}} [{(\mathbf{I} - \mathbf{\Sigma})}^{-1} - {(\mathbf{I} - \tilde{\mathbf{\Sigma}})}^{-1}] \tilde{\mathbf{U}}^T \\
 & \tilde{\mathbf{U}}[\mathbf{I} - (\mathbf{I} - \mathbf{\Sigma}){(\mathbf{I} - \tilde{\mathbf{\Sigma}})}^{-1}]\tilde{\mathbf{U}}^T.
\end{aligned} $

Finally, we have,
$$\BigS \tilde{\BigS}^{-1} = (\mathbf{L}\tilde{\mathbf{U}})[(\mathbf{I} - \mathbf{\Sigma}){(\mathbf{I} - \tilde{\mathbf{\Sigma}})}^{-1}] {(\mathbf{L}\tilde{\mathbf{U}})}^{-1} .$$
Q.E.D.
\end{proof}

It follows from this theorem that if we compute the top $\ell$ eigenvectors of 
$\mathbf{R}$ matrix, we can use these eigenvectors and $\BigMhard$ to efficiently 
solve the system in Equation~\eqref{eq:SKatz}. 
This is because, 
in the iterative solution of~\eqref{eq:SKatz}, we multiply $\mathbf{k_2}$ vector 
by a matrix that contains $\ell \times \ell$ identity matrix on top instead of $\BigS$ matrix at each iteration. From the theorem, we can approximate the first $\ell$ part of inverse of $\BigS$ via low rank and since this $\ell \times \ell$ upper part of inverse of $\BigS$ is already computed in the preproccessing phase, we automatically use precomputed part in the iterative computation of ~\eqref{eq:SKatz}.Setting the iterative process this way, we eliminate  $\ell \times \ell$ computation from equation ~\eqref{eq:SKatz}.
\begin{algorithm}
\caption{The Preprocessing Phase}\label{alg:preprocess}
\begin{algorithmic}[1]
\Procedure{Preprocess}{$\boldG, \alpha,k$}
\State Construct $\BigM \leftarrow (\Identitymatrix - \alpha \boldG)$
\State Use minimum-vertex-seperator graph partitioning on $\BigM$ to partition it into $\BigMeasy, \BigMleft, \BigMright, \BigMhard$ \cite{karypis1998fast,karypis1998parallel}
\State Decompose $\BigMeasy$ into $\mathbf{L_{1}}$ and ${\mathbf{L}_{1}}^T$ using Cholesky factorization and invert $\mathbf{L}_{1}$ and ${\mathbf{L}_{1}}^T$
\State Decompose $\BigMhard$ into $\mathbf{L}$ and ${\mathbf{L}^T}$ using Cholesky factorization
\State Create $\residualArnoldi^{(0)}$ as normalized random vector
\State Compute $\left[\tilde{\mathbf{U}}, \tilde{\mathbf{\Sigma}}\right]\ = Lanczos(\mathbf{R} \residualArnoldi^{(0)}, k,\residualArnoldi^{(0)} )$ \cite{demmel1997applied} \Comment{As matrix-vector product}
\EndProcedure
\end{algorithmic}
\end{algorithm}
\subsubsection{The \algo\ Algorithm}
In this section, we outline our algorithm for Katz-based network proximity computation. In ``offline" preprocessing Algorithm~\ref{alg:preprocess}, we first construct $\BigM$ and use \textit{PartGraphRecursive} in {\sc Metis} to partition
$\BigM$ in such a way that $\BigMeasy$ is a sparse block diagonal matrix, and $\BigMhard$ is
dense but smaller~\cite{karypis1998fast,karypis1998parallel}. Next, we reorder the entries
of partitioned matrix $\BigM$ based on an approximate minimum degree ordering
(AMD)~\cite{amd1,amd2}. After reordering entries of $\BigM$, we invert the  Cholesky factorization of sparse block-diagonal matrix $\BigMeasy$ and obtain ${\mathbf{L}_{1}}^{-1}$
and ${\mathbf{L}_{1}}^{-T}$. Then, we construct $\tilde{\mathbf{U}}$ and $\tilde{\mathbf{\Sigma}}$ via \textit{Lanczos} procedure \cite{demmel1997applied} without forming $\mathbf{R}$ matrix. Subsequently, we form matrices for $\Tilde{\BigS}$ and store the resulting values and matrices
into an index to use them in query processing phase of our algorithm, \algo\ .
\begin{algorithm}
\caption{The \algo\ Algorithm}\label{alg:lrcKatz}
\begin{algorithmic}[1]
\Procedure{LRC-Katz}{}
\State Partition vector $e_q$ into $e_1$ and $e_2$ for query, $q$
\State Construct $b_1$, $b2$, and $f = (g_2 - \BigMleft \BigMeasy^{-1} g_1)$
\State Create ${k_2}^{(0)}$ as normalized random vector
\State Set $i=0$, $r^{(i)}=f-\BigS {k_2}^{(k)}$, $s=\BigS r^{(i)}$, 
     $p=\Tilde{\BigS}\backslash s^{(i)}$, $y^{(i)}=\Tilde{\BigS}\backslash r^{(i)}$ \Comment{ $\BigS $ and $\Tilde{\BigS} $ are used as matrix-vector product}
\State  $\gamma^{(i)}=y^{(i)\mathrm{T}}s^{(i)}$
\If {$\gamma^{(k)}\le\epsilon$} \label{lst:line:gotoline}
\State  $k_2={k_2}^{(i)}$ and terminate
\EndIf
 \State  $q^{(i)}=\BigS p^{(i)}$
  \State  $\alpha^{(i)}=\frac{\gamma^{(i)}}{\|q^{(i)}\|^2}$
 \State  ${k_2}^{(i+1)}={k_2}^{(i)}+\alpha^{(i)}p^{(i)}$
  \State  $r^{(i+1)}=r^{(i)}-\alpha^{(i)}q^{(i)}$
 \State  $s^{(i+1)}=\BigS ^{(i+1)}$
\State  $y^{(i+1)}=\Tilde{\BigS} \backslash r^{(i+1)}$
\State  $\gamma^{(i+1)}=y^{(i+1)\mathrm{T}}s^{(i+1)}$
 \State  $p^{(i+1)}=\Tilde{\BigS} \backslash s^{(i+1)}+\frac{\gamma^{(i+1)}}{\gamma^{(i)}}p^{(i)}$
 
  \If{$i<i_{\text{max}}$}
  \State  $i\gets i + 1$ and go to line~\ref{lst:line:gotoline}
  \EndIf
  \State Compute $ k_1 \leftarrow {L_1}^{-1} ({L_{1}}^{-T}(g_1 - \BigMright k_2))$ and merge $k_1$ and $k_2$ as Katz-vector
\EndProcedure
\end{algorithmic}
\end{algorithm}

In the query phase of Katz-based proximity, for a given query node, $q$. We first construct the identity vector
$e_{q}$ and reorder the entries of $e_{q}$ using the same
ordering of $\BigM$. Subsequently, we divide $e_{q}$ into two parts,
$e_{q} = \begin{bmatrix}
 {e_q}_1 \\ {e_q}_2
\end{bmatrix}
$, based on the partition of $\BigM$ and set $b_1 =  {e_q}_1 - {(\Identitymatrix - \alpha \network)} {e_q}_1 $ and $b_2 = {e_q}_2 - {(\Identitymatrix - \alpha \network)} {e_q}_2 $. Next, we use the indexed matrices and $\Tilde{\BigS}$ to compute $k_2$ in equation ~\eqref{eq:SKatz}, the lower part of solution of
the linear system, with Conjugate Gradient method. Here, the $\Tilde{\BigS}$ serves as preconfitioner of Conjugate Gradient to refine norms of eigenvectors of $\BigS$. Finally, using $k_2$
and the indexed matrices, we compute $k_1$, the upper part of solution of the linear
system. We then merge the entries in $k_1$ and $k_2$ and return the resulting merged vectoras Katz-based network proximity vector.

\subsection{Efficient Processing of Link Prediction Queries via Katz Proximity} 

Link prediction can be defined as the problem of predicting the links that are likely to emerge/dissappear in the future, given the
current state of the network ~\cite{LinkLiben}. Various topological measures were broadly examined by Liben et al.~\cite{LinkLiben}, as unsupervised link prediction features. One can classify these measures into two categories: Neighborhood-based measures (local) and path-based measures (global). Clearly, The former are cheaper to compute, however, such local measures treat the network as a ``bag of interactions” rather than a true network, since they do not take into account the potential flow of information across the network through indirect paths, yet the latter (global) are more effective at link prediction since they account for the flow of the information through the indirect paths ~\cite{FastKatz, coskunLink}.

\begin{figure}[t]
	\centering
	\scalebox{0.4}{\input{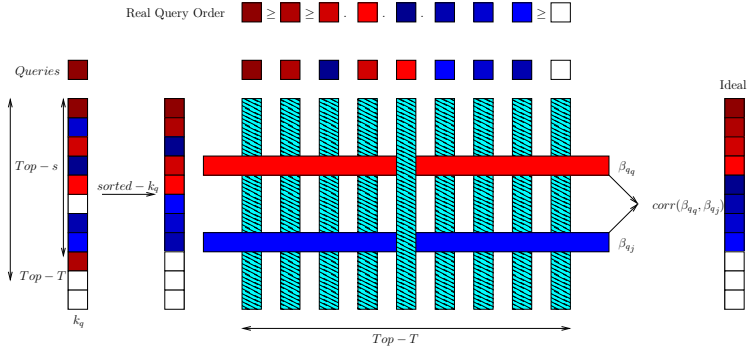}}
	\caption{{\bf Flowchart illustrating the proposed framework for \algoLink.} In the uppermost part of the Figure shows the assumed color strength, the darkest red is the strongest. From left to right, we start a query node $q$ and compute its Katz-based scores as $k_q$ and sort them. Assume that $s =8$ and $T= 9$, For $T$ nodes, we compute Katz-based scores of them and store as cyan matrix. Then, for row $q$ and $j$, we have ${\beta_q}_q$ and ${\beta_q}_j$ $9\times 1$ vectors. To measure the global closeness of $j$ to $q$, we take correlation of ${\beta_q}_q$ and ${\beta_q}_j$ and approximate to the rightmost vector.}
\label{fig:LinkPredictionFigure}
\end{figure}
In the context of disease gene prioritization, as an application of link prediction problem, global measures are also shown to be significantly more effective than local measures \cite{navlakhaGlobal}. However, these global measures are favor high-degree nodes over nodes with lower degree ~\cite{ertenDada}. To alleviate this problem, Erten et al. ~\cite{ertenGlobal} proposed Pearson Correlation based global method that assess the closeness of nodes in the network by comparing the views of the network from the perspective of nodes. More precisely, the authors' approach was to compute all proximity vectors and store them as a matrix for all nodes, as \textit{Katz matrix} $\Katzmatrix$ in equation ~\eqref{eq:KatzMatrix}, and take the row-wise Pearson Correlation as a new global proximity measure ~\cite{ertenGlobal}. Later, it has been shown that Pearson Correlation based approach's performance ~\cite{ertenGlobal} for the link prediction problem can adversely affected by high dimensionality, sparsity of the proximity vector ~\cite{coskunLink}. To alleviate this high dimesionality problem, Coskun et al. ~\cite{coskunLink} proposed two effective dimensionality reduction techniques. However, these techniques requires to compute all proximity matrix and store it in the memory. Clearly, we neither want to compute the \textit{Katz matrix}, $\Katzmatrix$ nor hold it in the memory.

Here in this paper, we propose an alternative algorithm, called \algoLink\, for Katz-based proximity measure that takes sparsity into account without forming the \textit{Katz matrix}. The idea behind \algoLink\ is as follows: For a given query node $q$ and an integer $s$, we compute Katz-based proximity, $k_q$ via \algo\ and sort the scores of $k_q$. Afterward, we take Top-s nodes corresponding to the highest $s\ $ scores in $k_q$. Then, we compute Katz-based proximity of top-T nodes corresponding top-T the highest scores in $k_q$. This way, we aim at treading those top-T nodes as modular nodes that see the query node as an important nodes in whole network. This approach enables us to approximate the dimensionality reduction approaches introduced by ~\cite{coskunLink} without forming $\Katzmatrix$ matrix. The workflow of the proposed framework is shown in Figure ~\ref{fig:LinkPredictionFigure}.

\begin{algorithm}
\caption{\algoLink }\label{alg:sparseKatz}
\begin{algorithmic}[1]
\Procedure{Sparse-Katz}{}
\State Given $Query \ q, positive \ integer \ s \ and \ T$
\State Compute $k_q$ with \algo\ 
\State Sort $k_q$ in descending and take $Top-s\ $ nodes
\For{\texttt{$t = 1: T$}}
        \State \texttt{Compute $k_t$ with \algo\ }
        \State \texttt{Store $k_t$s as matrix}
      \EndFor
\State Take row-wise Pearson Correlation of $K_T$ matrix and form $Top-T\ $ list 
\State Change $Top-s \ $ list according the order of $Top-T$ list
\EndProcedure
\end{algorithmic}
\end{algorithm}

\begin{figure*}[t]
  \begin{center}
    \subfigure[]{\label{fig:Iter-a}
    \includegraphics[scale=0.28]{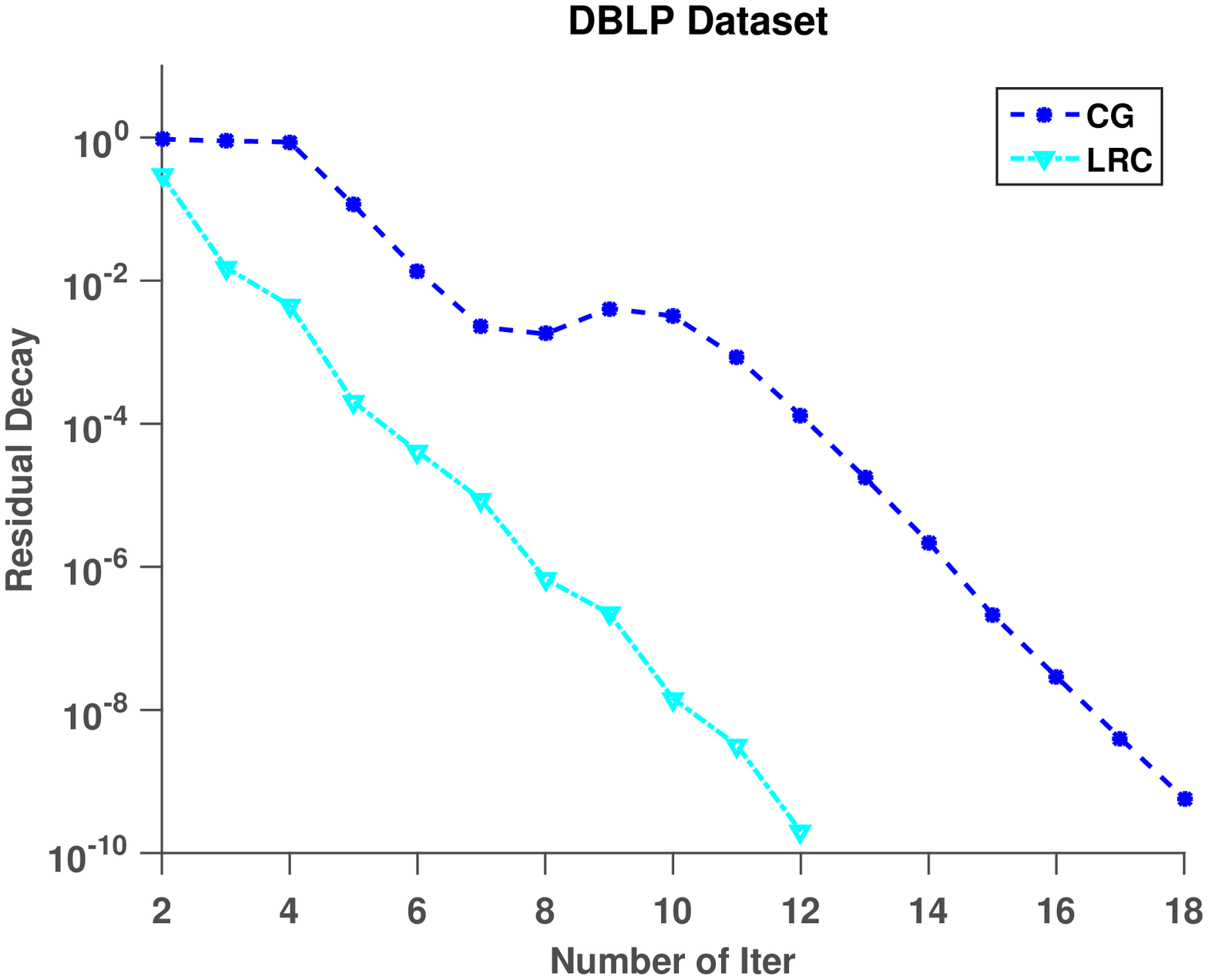}}\quad
    \subfigure[]{\label{fig:Iter-b}
    \includegraphics[scale=0.28]{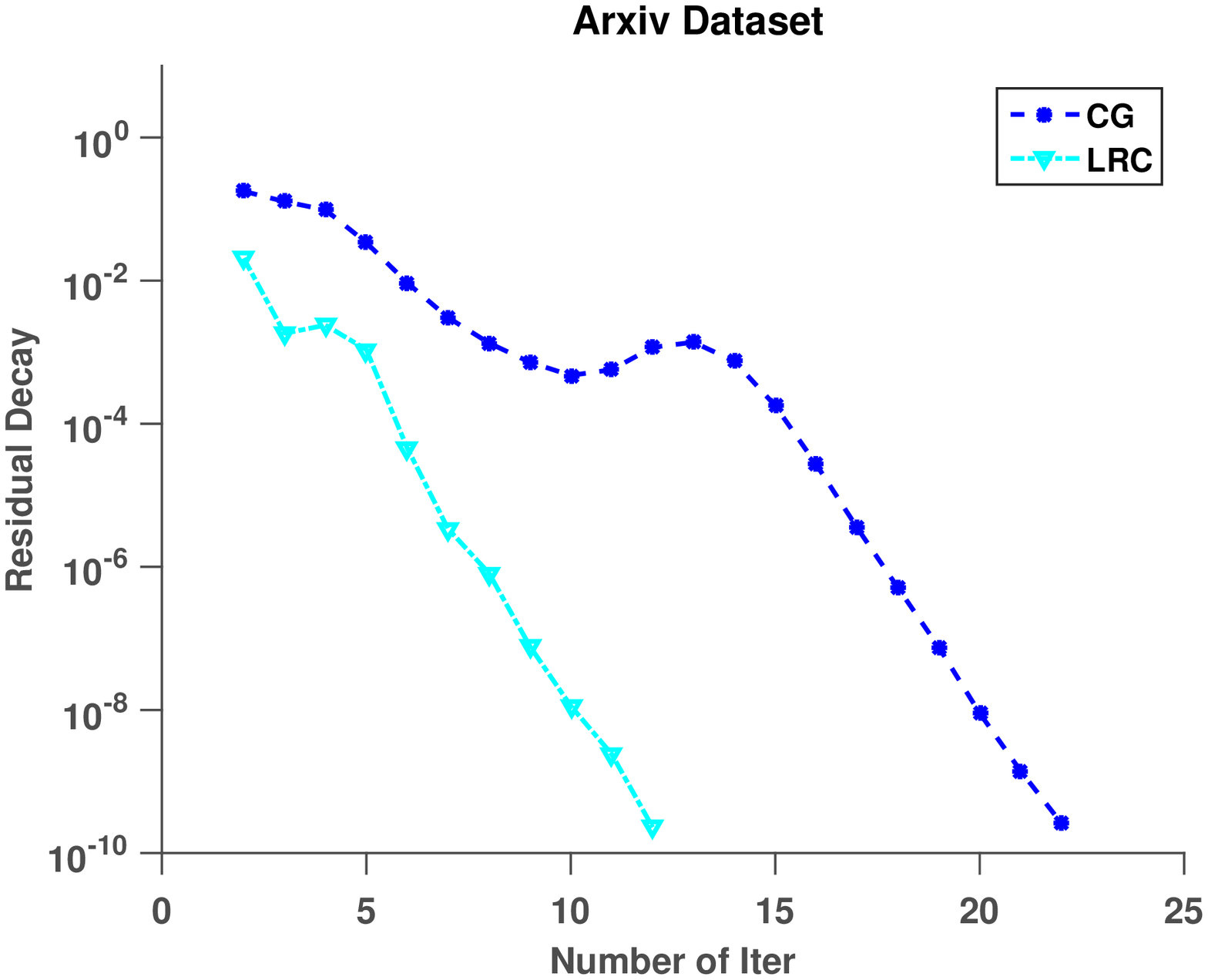}}\quad
     \subfigure[]{\label{fig:Iter-c}
     \includegraphics[scale=0.28]{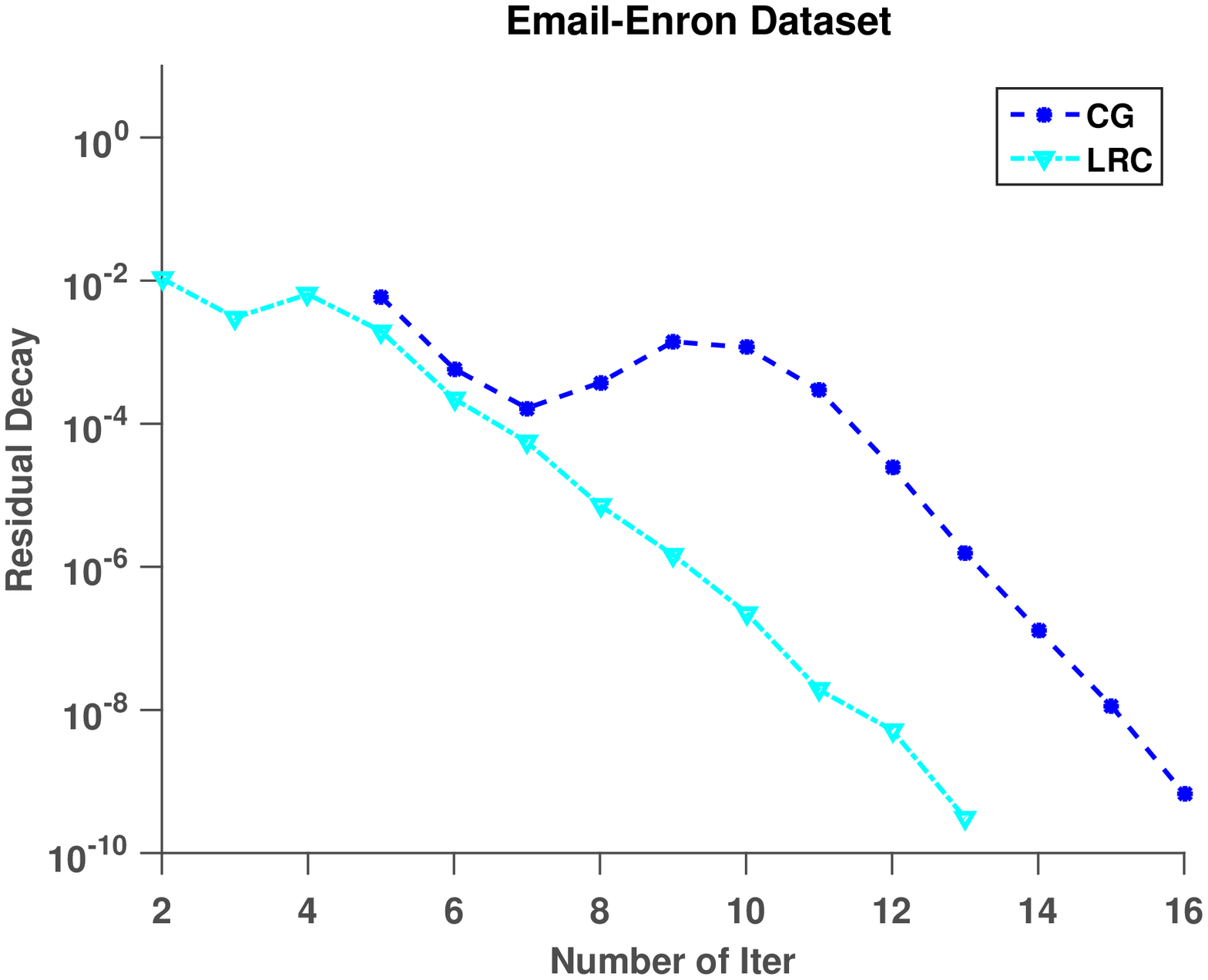}}\quad
     \subfigure[]{\label{fig:Iter-d}
     \includegraphics[scale=0.28]{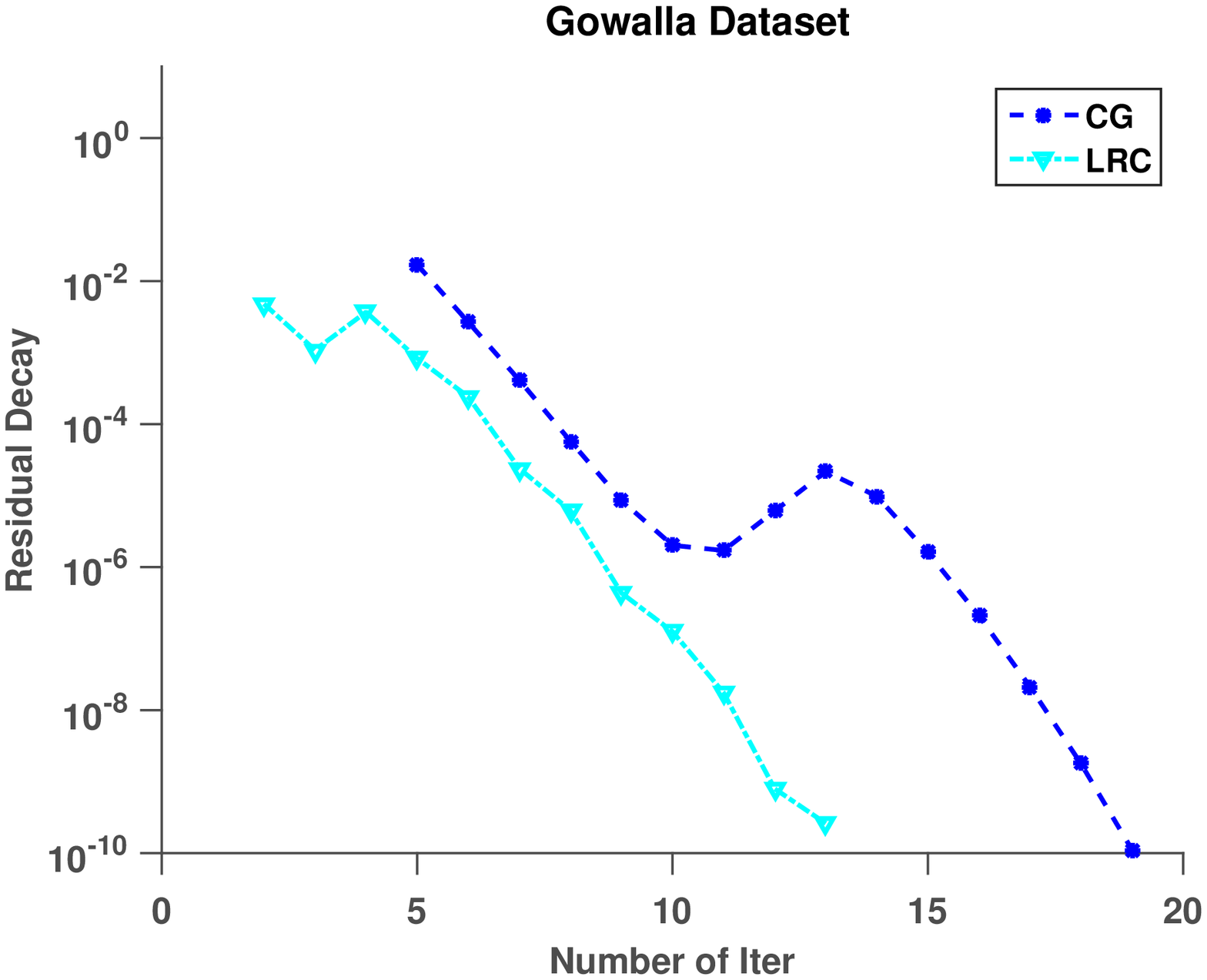}}\quad
     \subfigure[]{\label{fig:Iter-e}
     \includegraphics[scale=0.28]{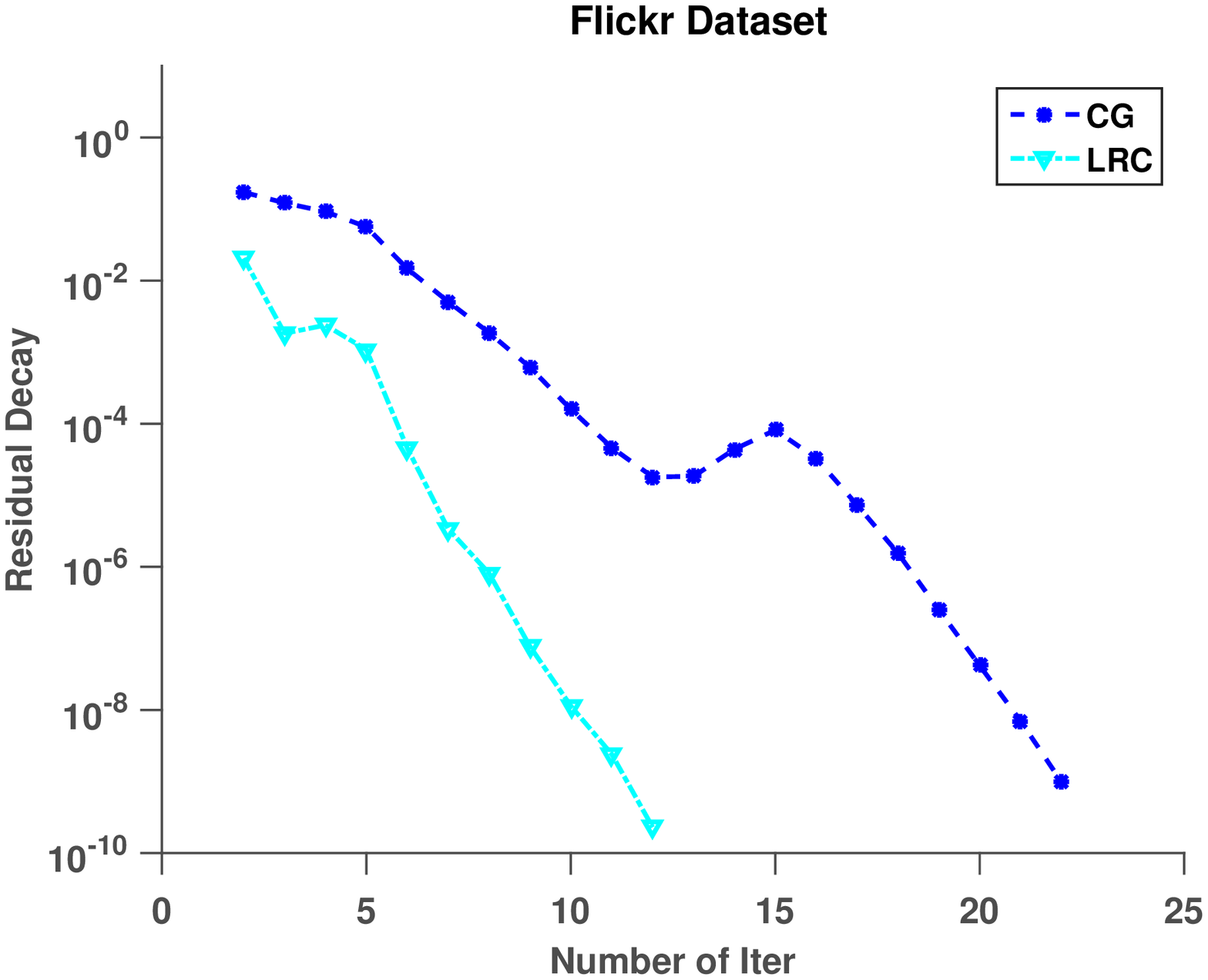}}
     \subfigure[]{\label{fig:Iter-f}
     \includegraphics[scale=0.28]{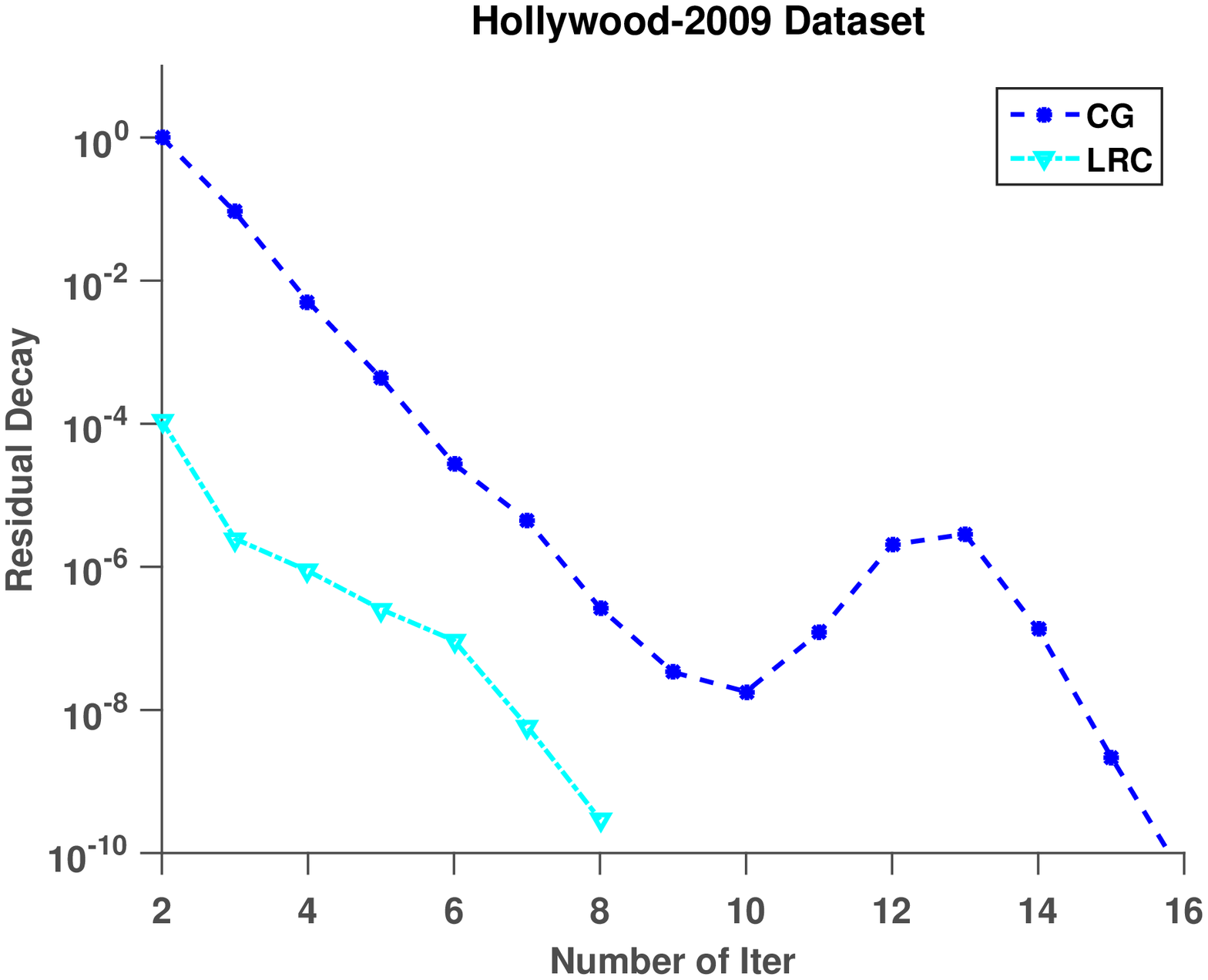}}
  \end{center}
 \caption{\textbf{ The number of iterations required for \algo\ and CG in computing Katz proximity scores}
 In these experiments, the reported numbers are the averages across $1000$ randomly chosen query nodes.}
  \label{fig:Iter}
\end{figure*}

\section{Experimental Results}
\label{sec:results}

\begin{table*}[tp]
\caption{Network data sets used in the experiments}
\label{datatable}
\centering
\resizebox{0.99\textwidth}{!}{%
\begin{tabular}{|c| c c c c|}
\hline
Network &  Number of Nodes &  Number of Edges &  Average Node Degree &  $\norm{G}_2$  \\
\hline\hline
 {\tt DBLP\_lcc} & 93,156 &  178,145 &  3.82  & 39.5753  \\
 {\tt Arxiv\_lcc} & 86,376  &  517,563 & 11.98  & 99.3319  \\
 {\tt Email-Enron} & 36,692 & 183,831  & 10.02 & 111.2871  \\
 {\tt Gowalla} & 196,591 & 950,327  & 9.67 & 169.3612  \\
 {\tt Flickr} & 513,969 & 3,190,452  & 12.41 & 663.3587  \\
 {\tt Hollywood-2009} & 1,139,905 &  113,891,327  &  99.13 &  2247.5591  \\ \hline\hline
 {\tt PPI \_Data} & 12,976  & 99,814  & 7.6916  & 94.4121\\
 {\tt DBLP\_2006-2008} &  179,266 & 765,346 & 4.2693 &  61.7765\\
\hline
\end{tabular}
}
\end{table*}

\begin{figure}[ht]
  \begin{center}
    \subfigure[]{\label{fig:Alpha-a}
    \includegraphics[scale=0.28]{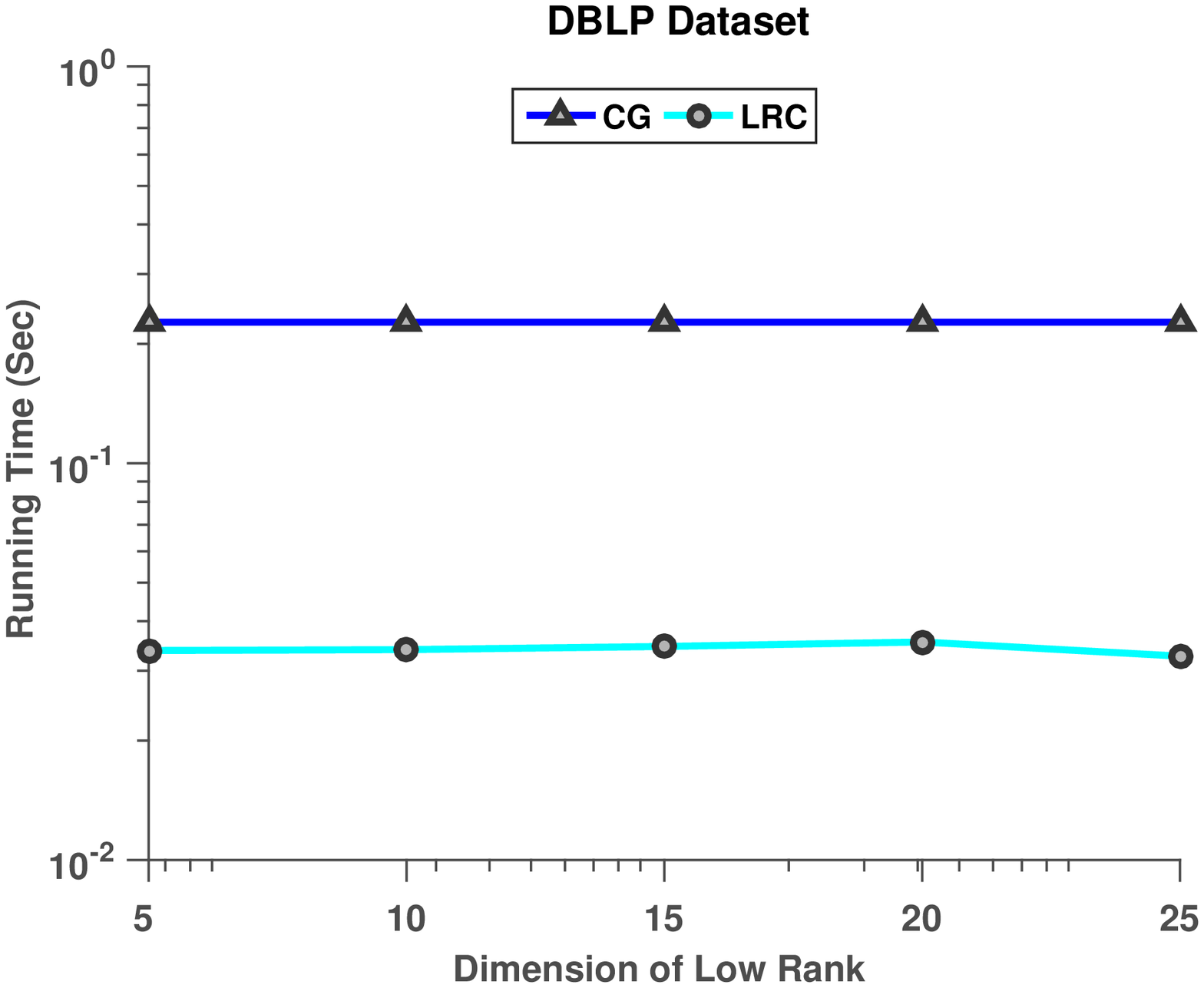}}\quad
    \subfigure[]{\label{fig:Alpha-b}
    \includegraphics[scale=0.28]{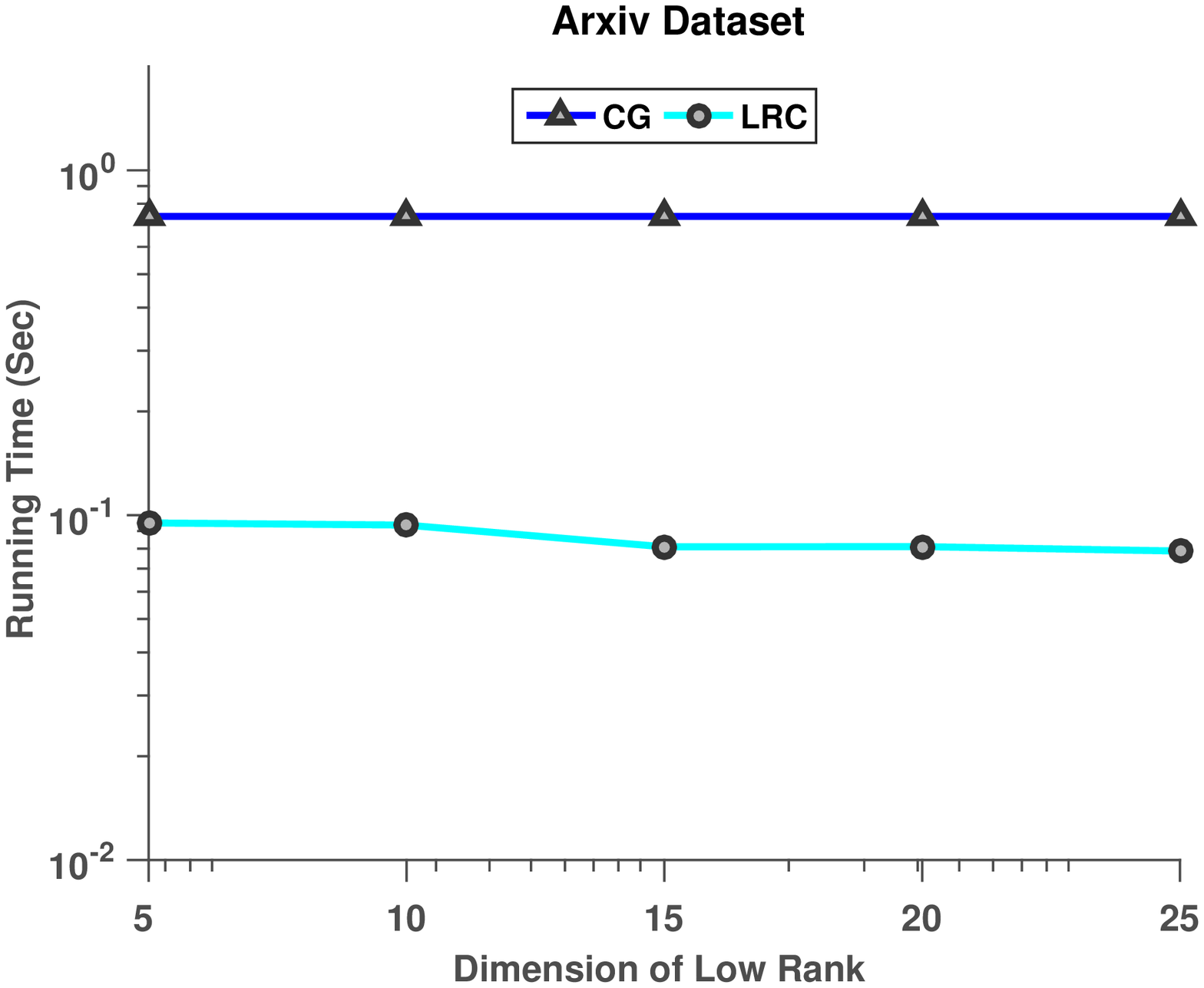}}\quad
     \subfigure[]{\label{fig:Alpha-c}
    \includegraphics[scale=0.28]{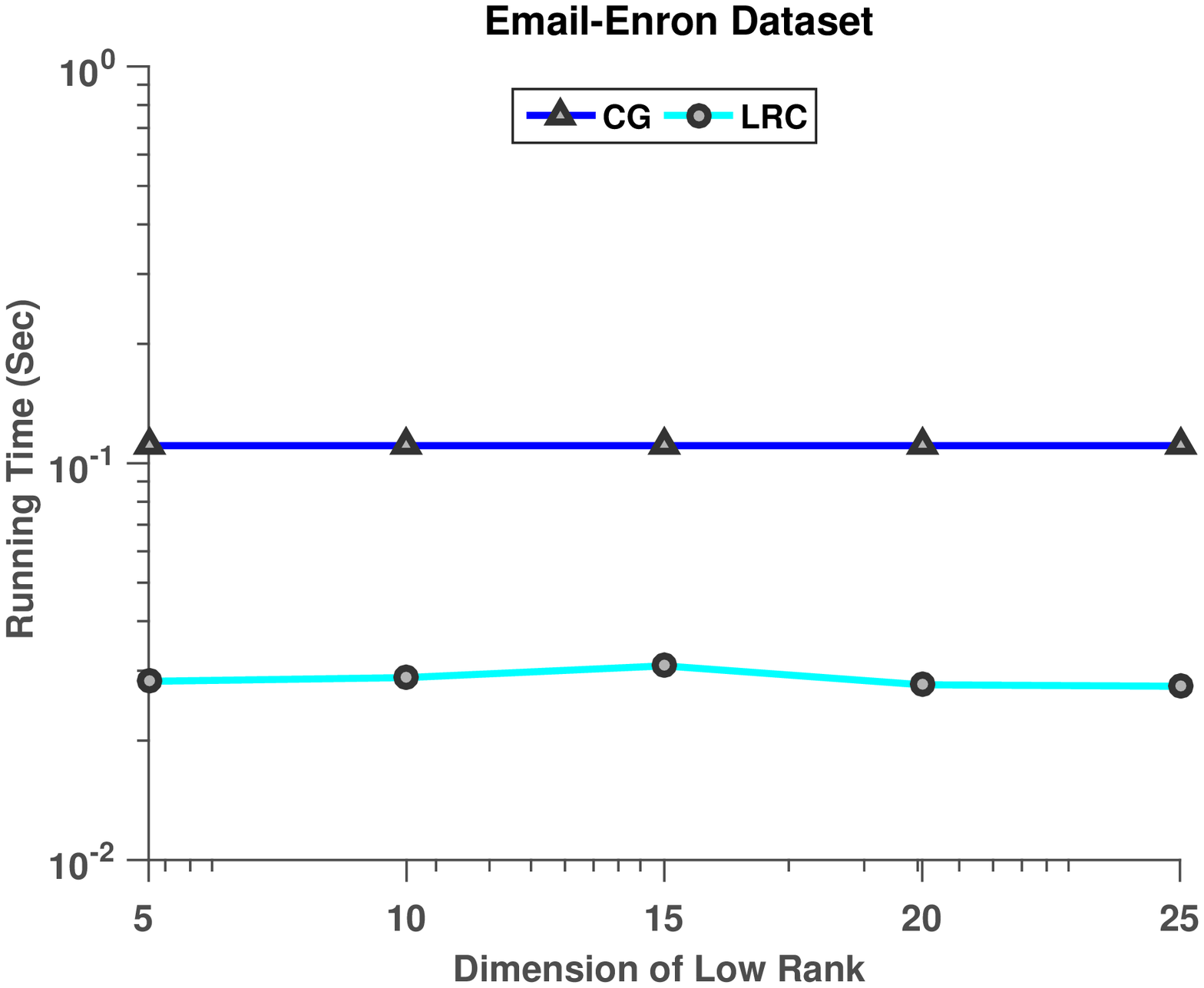}}\quad
     \subfigure[]{\label{fig:Alpha-d}
    \includegraphics[scale=0.28]{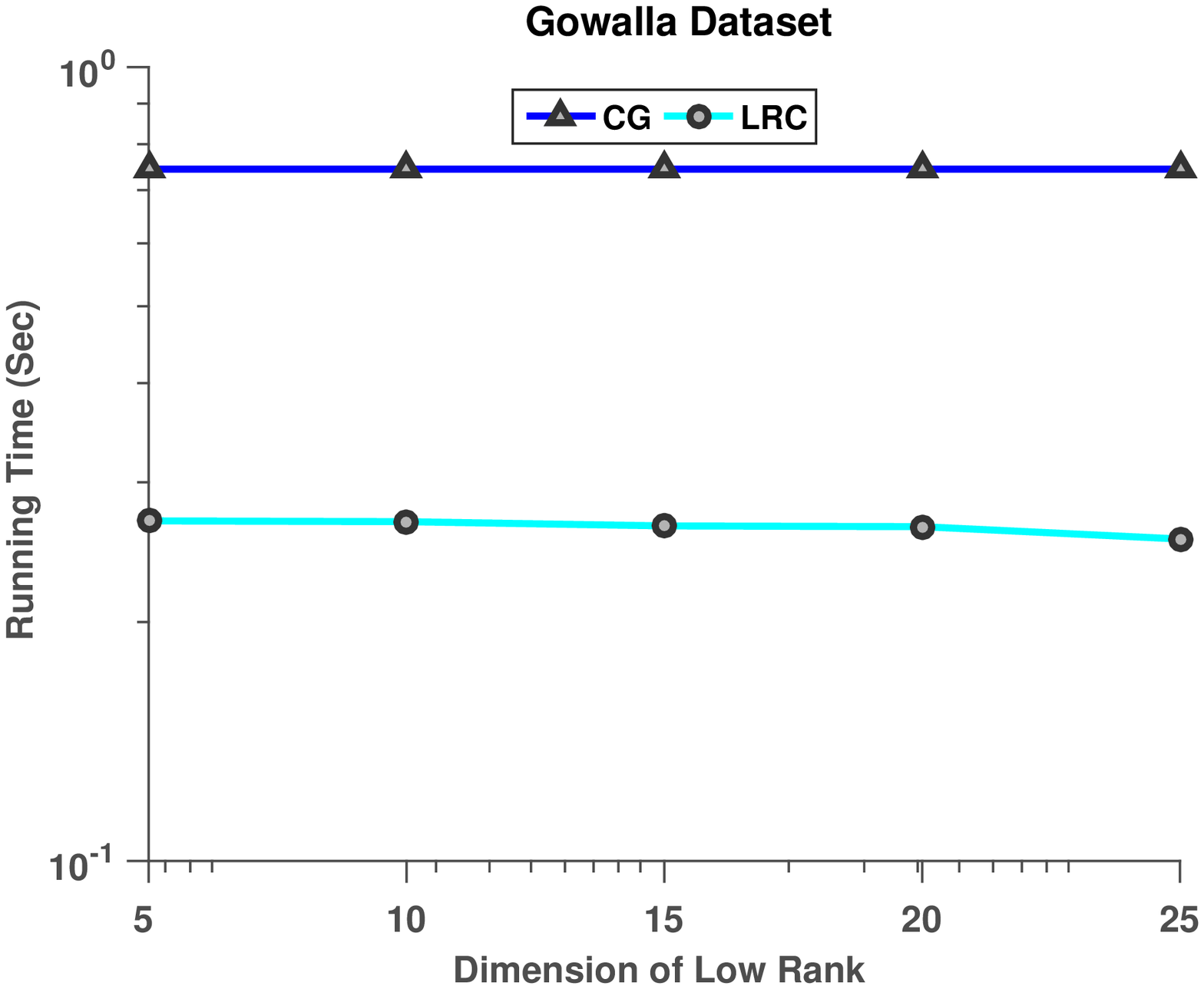}}\quad
     \subfigure[]{\label{fig:Alpha-e}
     \includegraphics[scale=0.28]{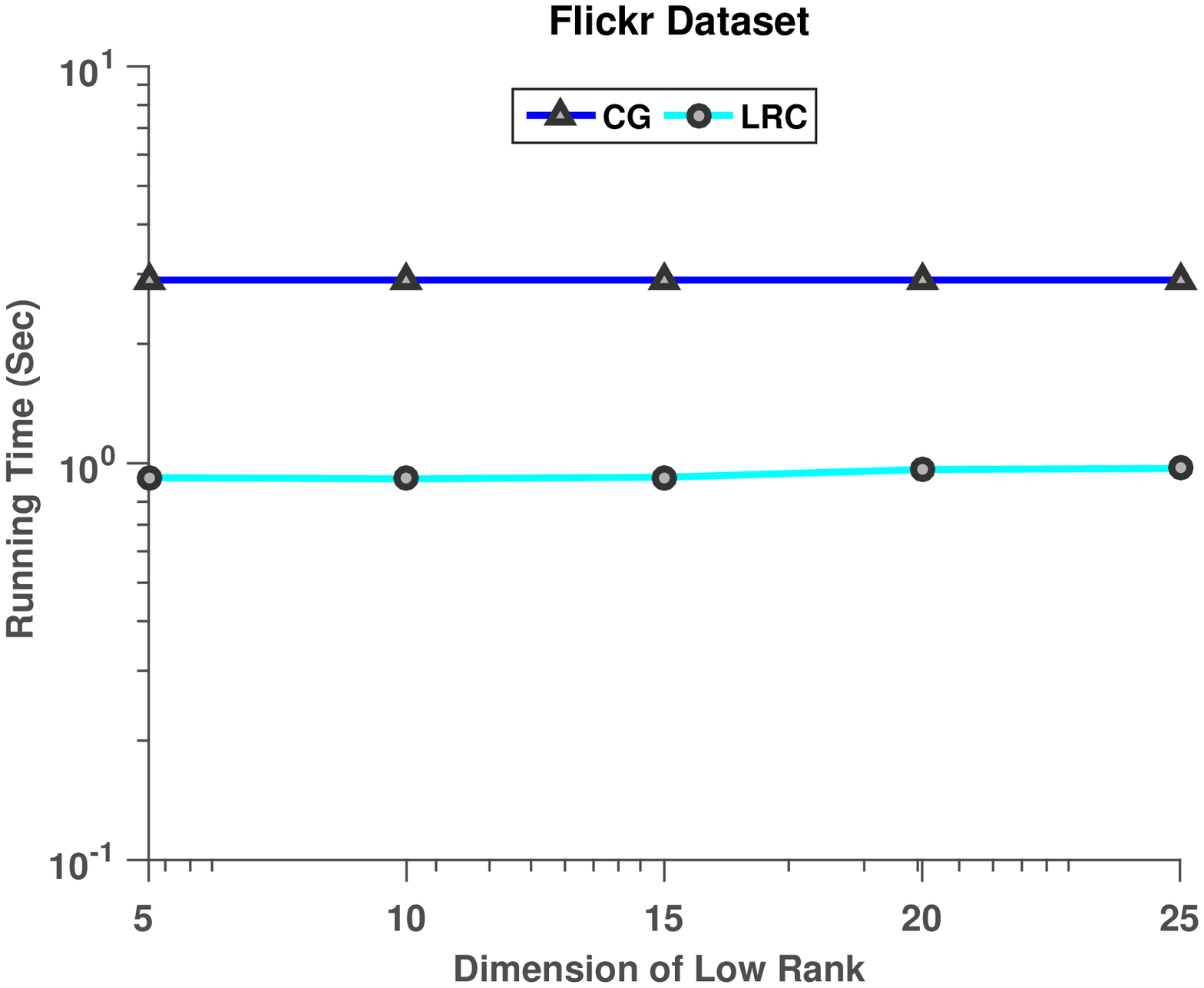}}
      \subfigure[]{\label{fig:Alpha-f}
     \includegraphics[scale=0.28]{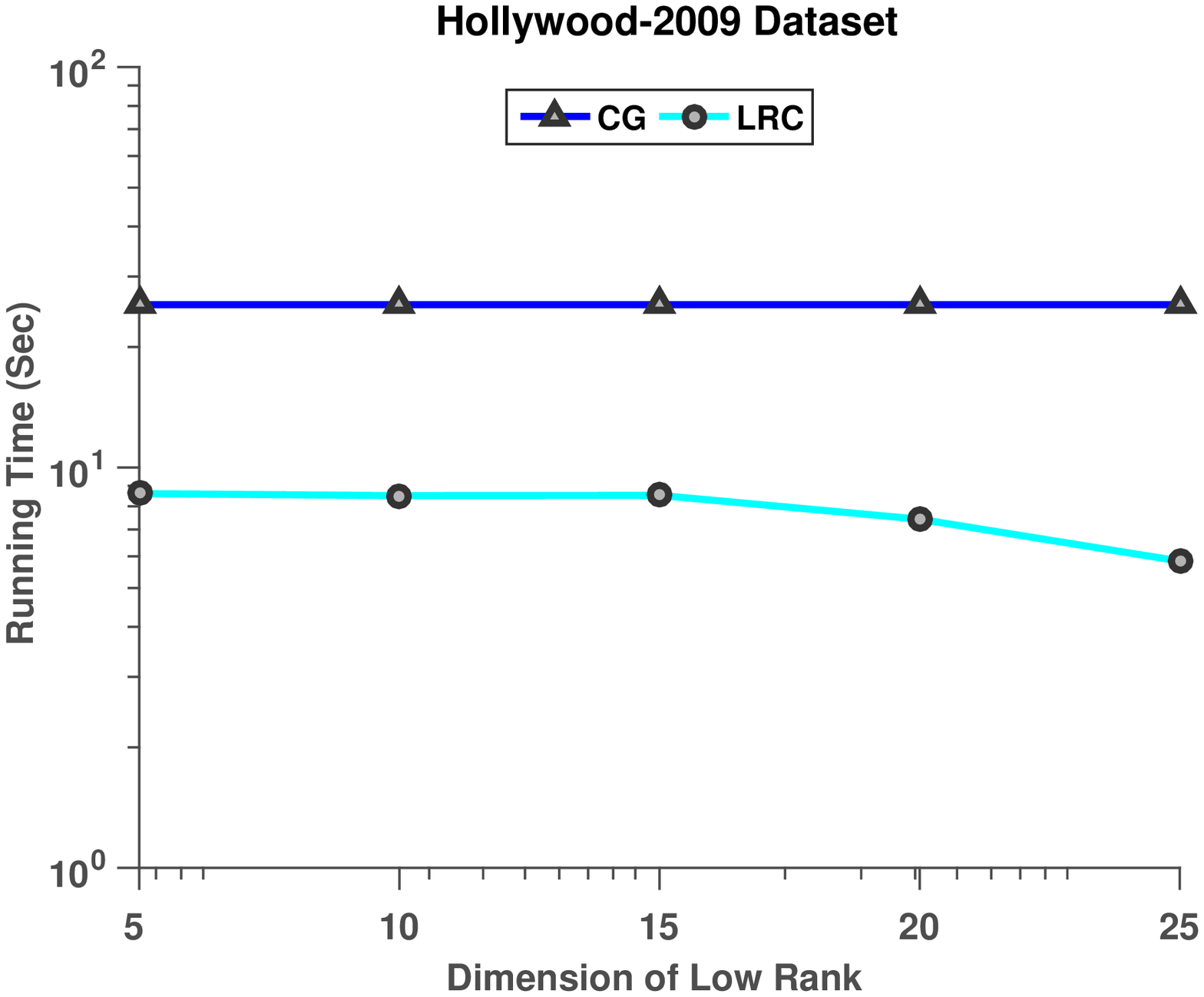}}
  \end{center}
 \caption{\textbf{Runtime in seconds of \algo\ and CG to process queries for Katz proximity as a function of $k$ ranging from $5$ to $25$}
In these experiments the reported numbers are the averages across $1000$ randomly chosen query nodes.}
  \label{fig:RunTime}
\end{figure}

\begin{figure}[ht]
  \begin{center}
    \subfigure[]{\label{fig:DBLP4}
    \includegraphics[scale=0.3]{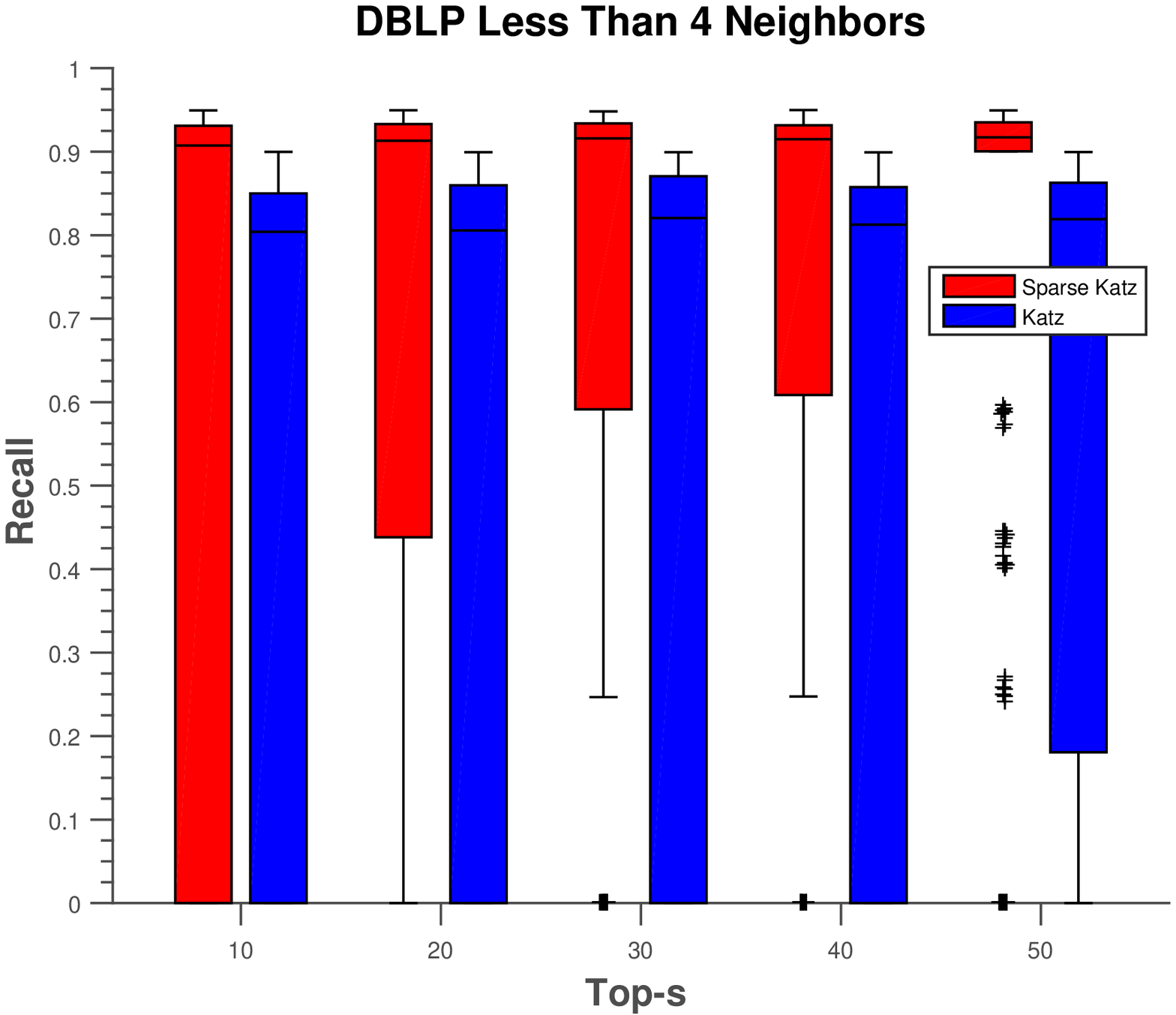}}\hfill
    \subfigure[]{\label{fig:DBLP4_10}
    \includegraphics[scale=0.3]{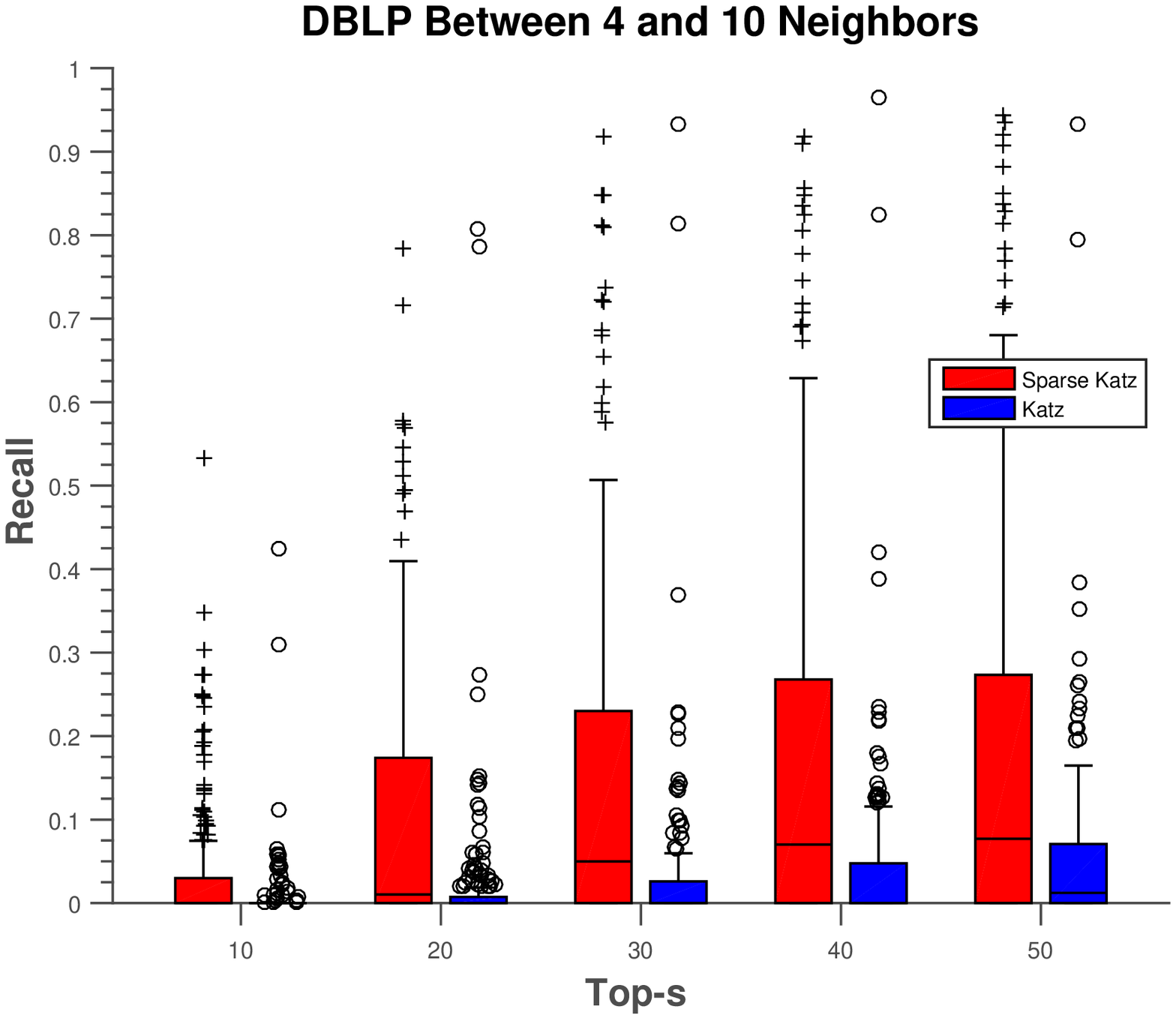}}\hfill
     \subfigure[]{\label{fig:DBLP10}\includegraphics[scale=0.3]{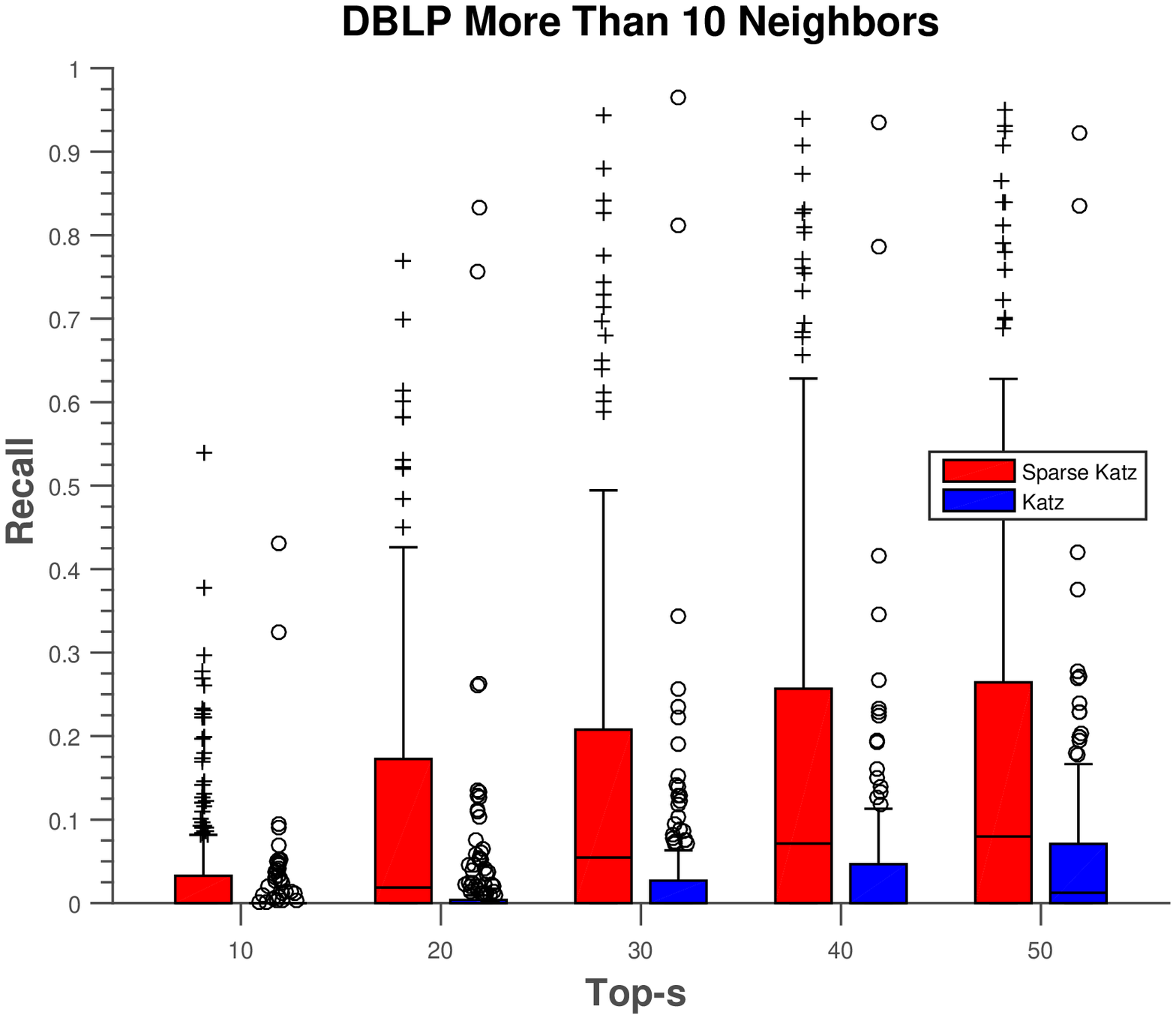}}
  \end{center}
 \caption{\textbf{Performance evaluation of \algoLink\ for  {\tt DBLP\_2006-2008}}
The performance of \algoLink\ link
prediction as
compared to Katz measure based link prediction on the  {\tt DBLP\_2006-2008} }
  \label{fig:DBLPLinkPredictionFig}
\end{figure}

\begin{figure*}[ht]
  \begin{center}
    \subfigure[]{\label{fig:PPI4}
    \includegraphics[scale=0.28]{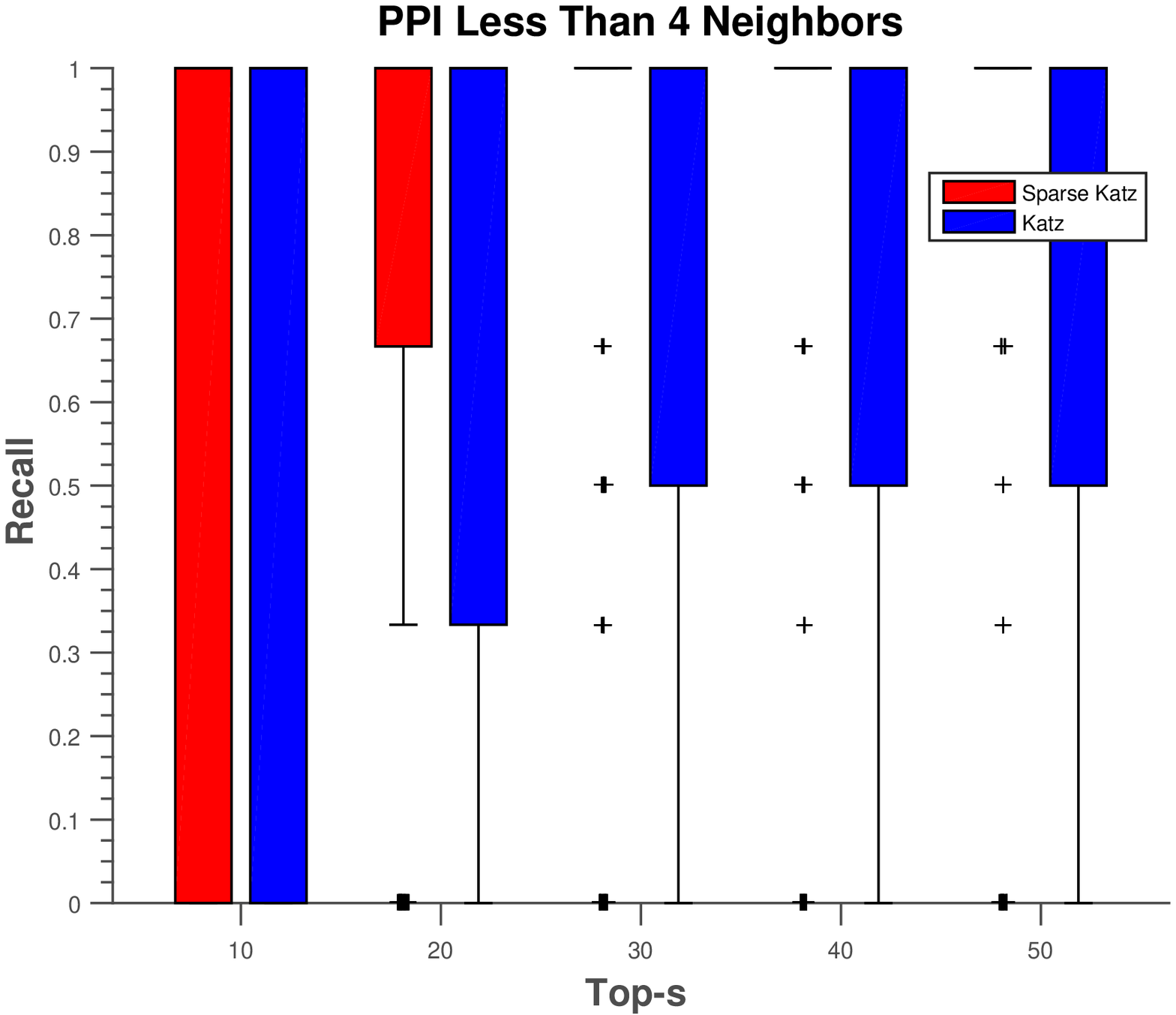}}\quad
    \subfigure[]{\label{fig:PPI4_10}
    \includegraphics[scale=0.28]{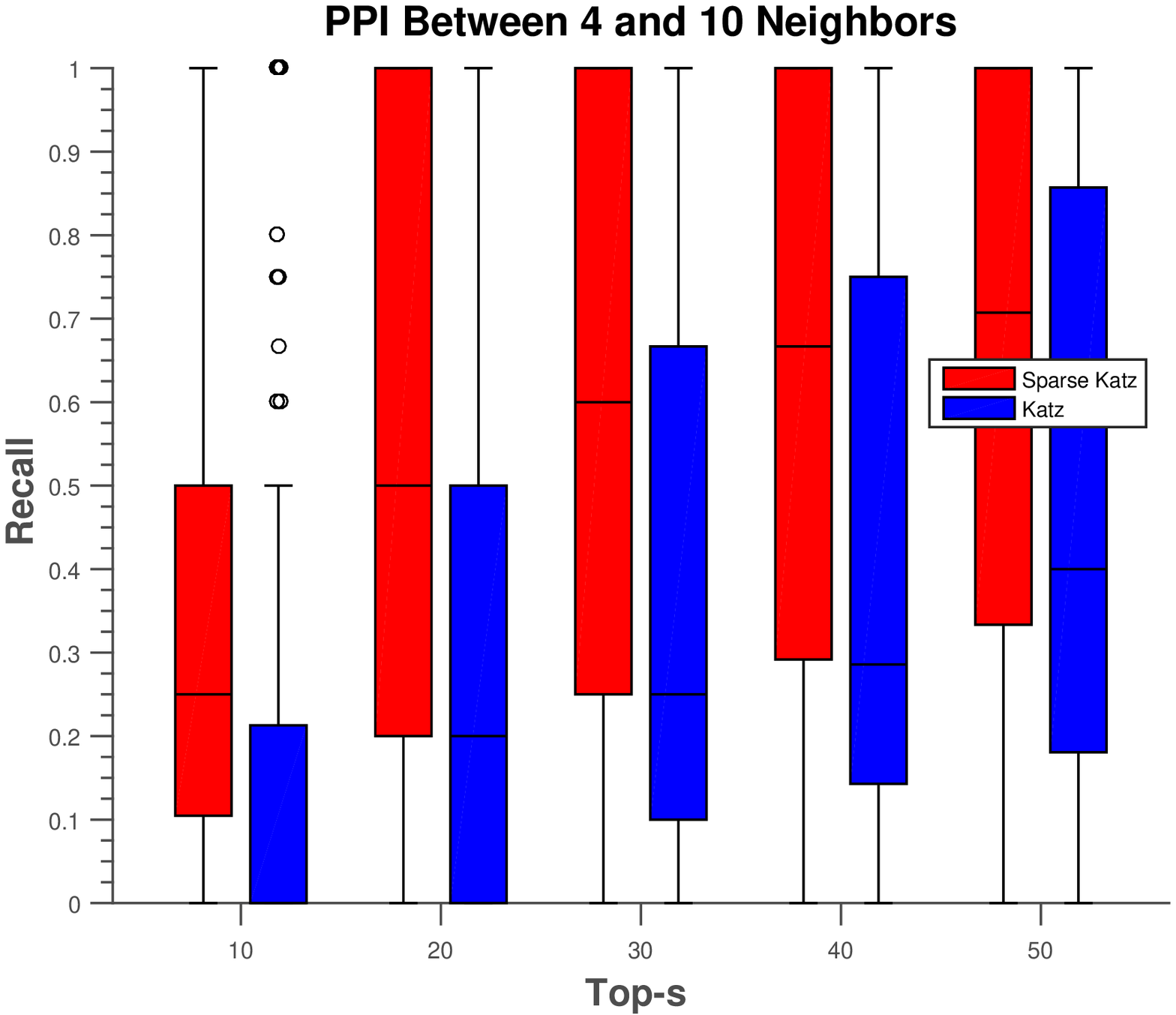}}\quad
     \subfigure[]{\label{fig:PPI10}\includegraphics[scale=0.28]{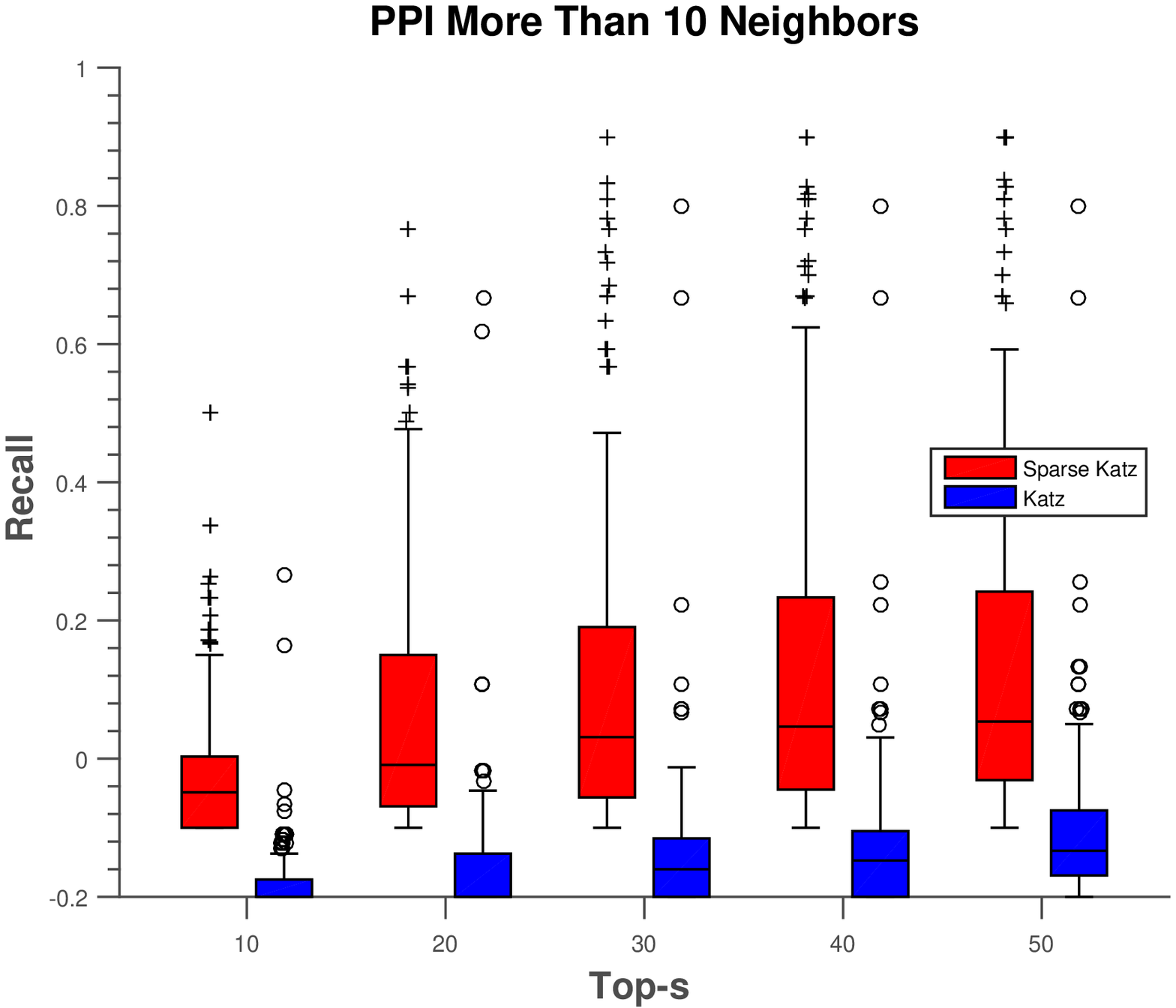}}\quad
  \end{center}
 \caption{\textbf{Performance evaluation of \algoLink\ for  {\tt PPI\_Data}}
The performance of \algoLink\ link
prediction as
compared to Katz measure based link prediction on the  {\tt PPI\_Data} }
  \label{fig:PPILinkPredictionFig}
\end{figure*}


In this section, we systematically first evaluate the runtime performance
the proposed algorithm, \algo\, in processing Katz-based proximity queries. As stated
in the preceding section, \algo\ is ``exact" in the sense  that it is guaranteed to
correctly identify Katz scores of a given query node. For
this reason, we focus on computational cost (measured in terms of number of iterations
and runtime) in our experiments for Katz-based proximity and compare \algo\ against another exact algorithm instead of top-$k$ based algorithms \cite{FastKatz}. We do not report the preprocessing time since it takes less than a few minutes even for the largest dataset.

We then evaluate the link prediction performance of \algoLink\ and compare it against Katz measure performance for link prediction. 

We start our discussion by describing the datasets and the experimental setup for Katz-based proximity. Subsequently, we compare the performance of
\algo\ in processing Katz-based network proximity queries with the fastest state-of-the-art algorithm, Conjugate Gradient method \cite{SaadBook}. Finally, we end our discussion by describing datasets for link prediction application of Katz-based proximity and comparing performance of our link prediction algorithm \algoLink\ against Katz score based link prediction algorithm.

\subsection{Datasets and Experimental Setup for Katz-based Proximity}

We use six publicly available real-world network datasets and two graphs we collected ourselves. We provide descriptive statistics of these eight networks in Table ~\ref{datatable}. The first six real-world networks, in our Katz-based proximity experiments. {\tt DBLP\_lcc} and {\tt Arxiv\_lcc} citation-based networks based
on publications databases, and {\tt Flickr} is social network, are publicly provided by \cite{FastKatz}. {\tt Email-Enron} email connection network at Enron and {\tt Gowalla} local social communication network come from SNAP collection \cite{leskovec2010signed}. Beside these datasets, we lastly use publicly available {\tt Hollywood-2009} ~\cite{boldi2011layered} Hollywood movie actor network for Katz-based network proximity experiments. The last two real-world network in Table ~\ref{datatable} refers the link prediction experiment datasets which collected by ourselves. We give detailed description of these two link prediction datasets in the subsequent sections.

In our implementation phase, for  CG algorithm, we use the Matlab implementation downloaded
from \cite{FastKatz}. We implement \algo\ and \algoLink\ in Matlab. We assess the performance of the algorithms for a fixed  damping factor, i.e for each dataset we use the $\alpha$s that defined as the hardest case  $\alpha = \dfrac{1}{\norm{\boldG}_2 +1}$ proximity queries ~\cite{FastKatz}. In practice, using such small $\alpha$ is recommended to fully utilize the information provided by the network \cite{coskun18}. For all experiments to assess the Katz proximity for the first six datasets, we
randomly select $1000$ query nodes and report the average of the performance figures for these $1000$ queries. In all Katz-based proximity experiments, $e_{q}$ is set to the identity vector  for node $q$. All of the experiments are performed on an Intel(R) Xeon(R) CPU E5-46200 2.20 GHz server with 500 GB memory.

\subsection{Runtime Performance for Katz-based Proximity}

The performance of \algo\ in comparison to Conjugate Gradient (CG) algorithm as a function of $k$, which is dimension of low rank eigenvectors, i.e number of top eigenvectors corresponding the top eigenvalues of $R$ matrix. We provide our analysis to six publicly available networks given in ~\cite{FastKatz,leskovec2010signed,boldi2011layered}. We first evaluate the performance of \algo\ against CG in terms of number of iteration for a fixed dimension $k=5$ for randomly chosen $1000$ query nodes. Resulting analysis is reported in Figure~\ref{fig:Iter} as an average iteration number. As it can be seen \algo\ significantly outperforms CG across all datasets.  We then assess the runtime performance of \algo\ in second with varying $k$, low rank dimensions and compare it against CG.As seen in the Figure ~\ref{fig:RunTime},
\algo\ outperforms CG, achieving more than 3-fold speed-up for
all networks. The performance of \algo\ improves as we increase the dimension, however, due to memory requirements of eigenvectors, we do not go beyond very fist eigvectors corresponding the top eigenvalues.

\begin{table*}[tp]
\caption{Positive Test Distribution of Link Prediction Datasets}
\label{datatableLinkPrediction}
\centering
\resizebox{0.99\textwidth}{!}{%
\begin{tabular}{|c| c c c |}
\hline
Network &  $\#$ of Edges $\left[1, 3\right]$ & $\#$ of Edges $\left[4, 10\right]$ & $\#$ of Edges $>10$ \\
\hline\hline
 {\tt DBLP\_2006-2008} & 194,474 & 97,237 & 16,873\\
 {\tt PPI\_Data} & 28,490  & 7,355  & 3,649 \\
\hline
\end{tabular}
}
\end{table*}

\subsection{Datasets and Experimental Setup for Link Prediction}
We test and compare the proposed algorithm \algoLink\ on two comprehensive datasets: one of which is real-world collaboration networks extracted from DBLP Computer Science Bibliography  \footnote{http://www.informatik.uni-trier.de/ley/db/}, and the other dataset consists of human protein-to-protein interaction (PPI) data obtained from the IntAct database \cite{orchard2013mintact}.

In the DBLP dataset {\tt DBLP\_2006-2008}, for training data, we consider authors who have published papers between 2006 and 2008. In this network, the authors are represented by nodes and there is a undirected link if two authors published at least one paper from 2006 to 2008 and, as test data, we use new co-author links that emerge between 2009 and 2010. In PPI dataset {\tt PPI \_Data}, for training data, we use proteins whose interactions are known in 2014, as test data, we use the proteins whose interactions are discovered in 2016, however, they have not been known in 2014. Here, proteins are used as nodes and interactions among them represent the edges. These datasets' descriptive static are provided in the last two rows of Table \ref{datatable}.

The objective of link prediction is to predict links that will emerge in the network in the future. For this reason, a  positive label in this setup refers to a new link that emerges
in the future version of a network, whereas a negative label refers to two nodes that remain unconnected in the future version. Since the real-world networks are highly sparse, the number of negative pairs is much larger than the number of positive pairs. Hence, predicting negative labels are relatively straightforward. Instead, the remaining challenge is to predict positive labels. Therefore, in our experiment, we only focus on predicting positive labels correctly and report a \textit{recall} evaluation rather than considering positive and negative labels together and reporting a \textit{recall-precision} evaluation. For this reason, we divide the positive labels into three categories to evaluate them . Namely, for DBLP data set, let $\netW$ denote the set of authors who published at least one paper in the testing interval [2009, 2010], but have not published together in the training interval ([2006, 2008]). We construct our positive node pairs from this set as follows:

\begin{itemize}
\item The positive test set $\PositiveSet$ is composed of $u, v \in \netW$ such that $u$ and $v$ published a paper between 2009 and 2010.
\item For DBLP network, $\PositiveSet$ consists of $308584$ added edges in between 2009 and 2010, i.e, $|\PositiveSet| = 308584$.
\item We divide $\PositiveSet$ into three categories as in summarized Table~\ref{datatableLinkPrediction}. 
\item For all experiments for {\tt DBLP\_2006-2008} dataset, we report the mean and the standard deviation of the performance figures
\end{itemize}
With similar logic above, we subsampled the $\PositiveSet = 39494$ sets in PPI network, {\tt PPI\_Data}, from 2016. In our all experiments, we use default hardest $\alpha$ values depending on  $\norm{G}_2$ values of  {\tt DBLP\_2006-2008} and {\tt PPI\_Data}.

\subsection{Link Prediction Performance for \algoLink}

The performance of \algoLink\ in comparison to Katz-based algorithm's link prediction application as a function of $s$, which is the top-s most probable edges that are computed by both algorithms, Katz and \algoLink. We provide our analysis to two networks given in the last two rows of Table \ref{datatable}. We first evaluate the performance of \algoLink\ against Katz in terms of number of \textit{recall} for all the positive set added to {\tt DBLP\_2006-2008} in between 2009 and 2010. Resulting analysis is reported in Figure~\ref{fig:DBLPLinkPredictionFig} as mean and standard deviation of the \textit{recall} values. As it can be seen \algoLink\ significantly outperforms Katz proximity based link prediction across all sub-sampled links.  We then assess the link prediction performance of \algoLink\ for {\tt PPI\_Data} dataset and compare it against Katz measure.As seen in the Figure ~\ref{fig:PPILinkPredictionFig},
\algoLink\ outperforms Katz measure drastically.
\section{Conclusion}
\label{sec:conclusion}
In this paper, we propose an alternate approach to accelerating
Katz- based network proximity queries. The proposed approach
is based on low rank correction of underlying partitioned linear systems of equation derived from Katz matrix. We show that our approach, \algo\ , significantly
decreases convergence times in practice on real-world
problems. Using a number of large real-world networks,
we show that \algo\ outperforms existing the fastest known method, Conjugate Gradient, significantly for wide ranges of parameter values. When integrated with the the fact that Katz scores distributes sparsely in the Katz vector, Chopper yields further improvement
in performance of link prediction task over state-of-the-art Katz measure.
Future efforts in this direction would include incorporation
of other proximity measures into our framework and their
applications. Furthermore, while \algo\ is an ``exact” algorithm
and our experiments focus on runtime performance
for this reason, there also exist approximate methods that
compromise accuracy for improved runtime. Constructing an approximate version of \algo can provide
further insights into the trade-off between runtime and accuracy
in the context of network proximity problems.


\bibliographystyle{spmpsci}      
\bibliography{mybib}   

\end{document}